\documentclass[lettersize,journal]{IEEEtran}
\makeatother
\usepackage{cite}
\usepackage{enumitem}
\usepackage{balance}
\usepackage[caption=false,font=normalsize,labelfont=sf,textfont=sf]{subfig}
\usepackage{amsfonts}
\usepackage{eqparbox}
\usepackage{amsmath}
\usepackage{amssymb}
\usepackage{bbm}
\usepackage{xcolor}
\usepackage{epsfig}
\usepackage{epstopdf}
\usepackage{algorithm}
\usepackage{algpseudocode}
\usepackage{graphicx}
\usepackage{dsfont} 
\usepackage{varwidth}
\usepackage{booktabs}
\usepackage{multirow}
\usepackage{diagbox}
\usepackage{makecell}
\usepackage{mathrsfs}
\usepackage[colorlinks,
            linkcolor=black,
            anchorcolor=black,
            citecolor=black]{hyperref}
            
\ifodd 1
\newcommand{\com}[1]{{\color{blue}#1}} 
\else

\newcommand{\com}[1]{}
\fi

\newtheorem{theorem}{\textbf{Theorem}}
\newtheorem{proof}{Proof}

\newcommand{\RNum}[1]{\uppercase\expandafter{\romannumeral #1\relax}}

\begin{document}
\title{Intelligent Channel Allocation for IEEE 802.11be Multi-Link Operation: When MAB Meets LLM}

\author{Shumin Lian, Jingwen Tong,~\IEEEmembership{Member,~IEEE,} Jun Zhang,~\IEEEmembership{Fellow,~IEEE,} and Liqun Fu,~\IEEEmembership{Senior Member,~IEEE}
\thanks{
This work was presented in part at IEEE WCNC 2025 \cite{lian2025dynamic}. Shumin Lian and Liqun Fu are with the School of Informatics, Xiamen University, Xiamen 361005, China (e-mails: smlian@stu.xmu.edu.cn; liqun@xmu.edu.cn). 
Jingwen Tong and Jun Zhang are with the Department of Electronic and Computer Engineering, The Hong Kong University of Science and Technology, Kowloon, Hong Kong (e-mails: eejwentong@ust.hk; eejzhang@ust.hk). The corresponding author is Liqun Fu.
}
}



\maketitle
\thispagestyle{empty}

\begin{abstract}
WiFi networks have achieved remarkable success in enabling seamless communication and data exchange worldwide. The IEEE 802.11be standard, known as WiFi 7, introduces Multi-Link Operation (MLO), a groundbreaking feature that enables devices to establish multiple simultaneous connections across different bands and channels. While MLO promises substantial improvements in network throughput and latency reduction, it presents significant challenges in channel allocation, particularly in dense network environments.
Current research has predominantly focused on performance analysis and throughput optimization within static WiFi 7 network configurations. In contrast, this paper addresses the dynamic channel allocation problem in dense WiFi 7 networks with MLO capabilities. We formulate this challenge as a combinatorial optimization problem, leveraging a novel network performance analysis mechanism. Given the inherent lack of prior network information, we model the problem within a Multi-Armed Bandit (MAB) framework to enable online learning of optimal channel allocations. Our proposed Best-Arm Identification-enabled Monte Carlo Tree Search (BAI-MCTS) algorithm includes rigorous theoretical analysis, providing upper bounds for both sample complexity and error probability. To further reduce sample complexity and enhance generalizability across diverse network scenarios, we put forth LLM-BAI-MCTS, an intelligent algorithm for the dynamic channel allocation problem by integrating the Large Language Model (LLM) into the BAI-MCTS algorithm. 
Numerical results demonstrate that the BAI-MCTS algorithm achieves a convergence rate approximately $50.44\%$ faster than the state-of-the-art algorithms when reaching $98\%$ of the optimal value. Notably, the convergence rate of the LLM-BAI-MCTS algorithm increases by over $63.32\%$ in dense networks. The code is available at \url{https://github.com/Lianshumin576/MLOandMCTS}.
\end{abstract}


\section{Introduction}
WiFi represents the world's most prevalent distributed wireless networking technology, managing approximately $70\%$ of global Internet traffic while supporting critical services, economic activities, and everyday conveniences \cite{crow1997ieee}. As wireless demands from emerging applications such as augmented reality, cloud gaming, and holographic communications continue to surge, the recently standardized IEEE 802.11be, known as WiFi 7, promises substantial enhancements in network throughput and latency reduction \cite{deng2020ieee}. WiFi 7 introduces transformative features including 320 MHz channel bandwidth, 4K-quadrature amplitude modulation, Multi-Link Operation (MLO), and enhanced Quality-of-Service (QoS) management \cite{chen2022overview}. Among these innovations, MLO emerges as a particularly promising technology, enabling STAtions (STAs) and Access Points (APs) to operate simultaneously across multiple channels and frequency bands, including 2.4 GHz, 5 GHz, and 6 GHz, thereby significantly boosting network throughput \cite{carrascosa2023wi}.

Despite MLO's considerable potential for enhancing network performance and reducing latency, efficient performance analysis and channel allocation in dynamic, dense WiFi 7 environments remain insufficiently explored, necessitating innovative approaches. MLO demonstrates distinct advantages over conventional techniques such as channel bonding \cite{deek2011impact} and multi-WiFi acceleration \cite{saputra2013aggregate}. Unlike channel bonding, which is constrained to adjacent channels, MLO enables connections across disparate frequency bands. Furthermore, MLO performs data aggregation at the Media Access Control (MAC) layer, contrasting with the higher-layer data management of multi-WiFi acceleration. These unique characteristics underscore the necessity for novel performance analysis methodologies tailored to WiFi 7 networks. In addition, the evolution toward increasingly dense WiFi deployments presents new challenges. The simultaneous data transmissions enabled by MLO's dual- or tri-band radio capabilities frequently result in overlapping transmission collisions in dense networks, creating an urgent need for efficient dynamic channel allocation algorithms.

Current research predominantly addresses performance analysis and throughput optimization within static WiFi 7 network configurations \cite{lopez2022multi, bellalta2023delay}. However, these approaches overlook critical challenges in multi-AP deployments, where intensified channel contention and cross-link interference significantly impede MLO's effectiveness. The emergence of centralized control frameworks addresses these operational gaps by enabling comprehensive resource orchestration, essential for realizing MLO's full potential in dynamic environments. Recent advances in digital and hardware technologies have made such frameworks viable for consumer-grade applications. Zhang et al. \cite{zhang2024ieee} examined AP-STA pairing and link allocation in centrally coordinated WiFi 7 networks, proposing a proportional fairness algorithm. However, their optimization framework's reliance on static network assumptions, including deterministic channel allocation and periodic updates based on averaged data rates, results in substantial performance degradation in dynamic scenarios. This limitation emphasizes the critical need for dynamic methodologies incorporating online learning and real-time parameter adaptation.

\begin{figure}[!t]
\centerline{\includegraphics[scale=0.6]{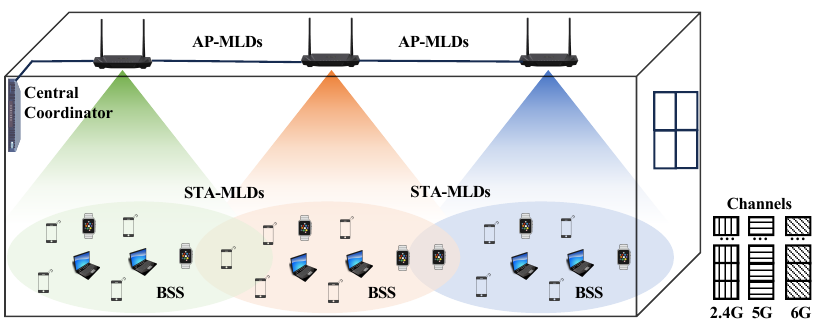}}
\caption{A dense WiFi 7 network comprising three Basic Service Sets (BSSs). Multi-Link Devices (MLDs) operate concurrently across multiple links spanning different channels and frequency bands (2.4, 5, and 6 GHz).}
\label{scenario}
\end{figure}

This paper investigates the dynamic channel allocation problem in dense WiFi 7 networks with MLO, as illustrated in Fig.~\ref{scenario}. We introduce a Continuous-Time Reversible Markov (CTRM) model \cite{jiang2009distributed} to accurately compute WiFi 7 network throughput by capturing the essential characteristics of the Carrier Sense Multiple Access (CSMA) protocol. Building upon this analytical framework, we formulate the channel allocation challenge as a combinatorial optimization problem. However, direct solution approaches face significant obstacles due to limited prior network information and the vast number of potential MLO configurations. Constructing comprehensive carrier-sensing graphs proves computationally prohibitive, necessitating efficient online learning algorithms to evaluate feasible solutions. Moreover, the solution space expands exponentially with varying MLO configurations across different AP-STA pairs.

To address these challenges, we reformulate the channel allocation problem within a Multi-Armed Bandit (MAB) framework, enabling online learning of optimal allocation strategies through efficient exploitation-exploration balancing \cite{bubeck2012regret}. Nevertheless, computational complexity remains a significant concern in MAB problems, particularly for latency-sensitive wireless applications. This complexity stems from two primary sources: the requirement for multiple explorations per arm to accurately estimate values, and the exponentially growing arm space. For instance, a network with six AP-STA pairs offering seven configurations each generates an arm space of $7^6 = 117,649$ possibilities. To enhance sample efficiency, we propose the Best-Arm Identification-enabled Monte Carlo Tree Search (BAI-MCTS) algorithm, synergistically combining MCTS and BAI techniques. MCTS excels in high-dimensional spaces by integrating Monte Carlo simulation with expansion mechanisms, constructing asymmetric search trees while efficiently balancing exploration and exploitation. To address MCTS's heuristic limitations and ensure solution quality, we incorporate BAI algorithms to provide guaranteed optimal channel allocations.

The vast arm space in MAB problems for dense WiFi networks presents an additional challenge. Recent breakthroughs in Large Language Models (LLMs) have catalyzed a paradigm shift in wireless network applications \cite{shao2024wirelessllm}, inspiring us to leverage LLMs for generating effective initial allocations in the BAI-MCTS algorithm, thereby reducing the exploration space. LLMs offer two key advantages: First, they rapidly integrate domain knowledge through In-Context Learning (ICL) \cite{dong2022survey}, enhancing algorithm generalization. Solutions from small-scale combinatorial optimization scenarios serve as exemplars, enabling LLMs to extract problem structures and generalize to larger-scale deployments, effectively bridging optimization-based and data-driven methodologies. Second, LLMs comprehend channel allocation procedures and conflict relationships among APs and STAs through Chain-of-Thought (CoT) reasoning \cite{chu2023survey}, facilitated by carefully designed prompt chains and task descriptions. These capabilities motivate our development of the LLM-BAI-MCTS algorithm for rapid, intelligent channel allocation in dense WiFi 7 networks.

The principal contributions of this work are as follows:
\begin{itemize}
    \item We address the dynamic channel allocation problem in dense WiFi 7 networks with MLO, introducing the CTRM model for network throughput calculation and formulating it as a combinatorial optimization problem. We subsequently model this as an MAB problem to facilitate online learning of optimal channel allocations.
    \item We develop the BAI-MCTS algorithm, which achieves an efficient exploitation-exploration balance. We derive theoretical upper bounds for both sample complexity and error probability, proving that BAI-MCTS converges to an $\epsilon$-optimal policy with probability exceeding $1-\delta$, where $\epsilon$ and $\delta$ represent small positive constants.
    \item To address large arm spaces in MAB problems, we introduce LLM-BAI-MCTS, an intelligent, LLM-assisted variant that significantly reduces the arm space while enhancing generalizability. This algorithm leverages CoT and ICL methods to transfer problem structures from small-scale to large-scale scenarios.
    \item Numerical evaluations demonstrate that BAI-MCTS achieves approximately $50.44\%$ faster convergence than the Dirichlet-Normal Gamma MCTS (DNG-MCTS) algorithm when reaching $98\%$ of optimal performance. Furthermore, LLM-BAI-MCTS exhibits convergence rates approximately $63.32\%$ faster than BAI-MCTS in dense network environments.
\end{itemize}

The paper is structured as follows: Section~\ref{SecRW} reviews related work. Section~\ref{SecSM} presents the system model, while Section~\ref{SecPF} formulates the dynamic channel allocation problem in WiFi 7 networks with MLO. The BAI-MCTS algorithm is detailed in Section~\ref{SecBA}, with its enhanced variant, LLM-BAI-MCTS, described in Section~\ref{SecLLM}. Section~\ref{SecNA} provides comprehensive algorithm evaluations. Finally, Section~\ref{SecCon} presents conclusions and future directions. For reference, Table~\ref{tab: Notation} summarizes the primary notation used throughout this paper.

\section{Related Work}\label{SecRW}
WiFi throughput analysis predominantly employs Bianchi and Markov models as foundational frameworks. Notable contributions by \cite{tong2021throughput,jiang2009distributed} demonstrate that constructing CSMA Markov chains and jointly analyzing steady-state probabilities with transmission rates enables analytical derivation of per-link throughput, which can be aggregated to comprehensively assess network capacity. However, extending these established models to accommodate MLO scenarios in WiFi 7 networks remains a largely underexplored research area.

Research on WiFi 7 networks with MLO has primarily bifurcated into two complementary domains: performance analysis and throughput optimization. The performance analysis domain encompasses comprehensive investigations of network behavior, including the coexistence dynamics between traditional Single-Link Devices (SLDs) and Multi-Link Devices (MLDs) \cite{lopez2022multi}, systematic exploration of novel features introduced in WiFi 7 \cite{chen2022overview}, and comparative evaluations of Single-Link Operation (SLO) versus MLO across diverse scenarios \cite{carrascosa2022experimental}. The throughput optimization domain concentrates on developing efficient transmission mechanisms. For instance, \cite{park2024adaptive} and \cite{zhang2024wifi} investigated optimal back-off window size design strategies to maximize network throughput. However, these contributions primarily target single AP configurations, limiting their applicability in complex multi-AP deployments.

Recent research endeavors have progressively shifted toward optimizing channel allocation through multi-AP collaboration frameworks. Ref. \cite{iturria2023channel} proposed an innovative parallel transfer reinforcement learning algorithm incorporating optimistic-weighted value decomposition networks to optimize channel allocation for MLO. Similarly, \cite{zhang2024ieee} introduced a sophisticated data-driven resource allocation algorithm, assisted by an AP controller, to maximize network throughput while maintaining fairness constraints. Furthermore, \cite{ali2023federated} developed a federated reinforcement learning framework for link activation, enabling neighboring Basic Service Sets (BSSs) to collaboratively learn optimal link allocation strategies for MLO. Nevertheless, these approaches predominantly address static WiFi 7 network configurations, neglecting the inherently dynamic nature of real-world channel conditions. Moreover, they inadequately address computational efficiency considerations, which prove critical for practical deployments. Consequently, developing efficient dynamic channel allocation methodologies for dense WiFi 7 networks with MLO remains an open research challenge.

The investigation of combinatorial optimization problems proves fundamental for designing and managing wireless networks, addressing critical challenges including channel allocation, network routing, and interference mitigation \cite{vesselinova2020learning}. However, solving these problems within wireless network contexts presents formidable challenges due to their inherent complexity, dynamic characteristics, and stringent low-latency requirements \cite{letaief2019roadmap}. Early research typically assumed complete network information availability and perfect solvability of combinatorial optimization problems from purely optimization perspectives \cite{lin2006tutorial}. As networks have evolved toward greater complexity and heterogeneity, researchers have increasingly adopted data-driven methodologies to learn optimal solutions \cite{zhang2019deep}. Nevertheless, these approaches typically demand substantial training data volumes, rendering them impractical for numerous real-world scenarios.
\begin{table}[t]   
  \centering  
  \caption{Main Notations} 
  \label{tab: Notation} 
  \begin{tabular}{cp{6.4cm}} 
    \hline 
    \noalign{\hrule height 0.6pt} 
    Notation & Description \\  
    \hline 
    $M$ & The number of BSSs or APs in the network \\  
    $N$ & The number of STAs in the network \\  
    $N_m$ & The number of STAs in the $m$-th BSS \\ 
    $\mathcal{T}(n,m)$ & The throughput of AP$_m$-STA$_n$ pair \\
    $L_t{(n,m)}$ & The selected MLO configuration of AP$_m$-STA$_n$ pair at time slot $t$ \\
    $\mathcal{L}(n,m)$ & All feasible MLO configurations of AP$_m$-STA$_n$ pair \\
    $\mathcal{J}(m)$ & All feasible MLO configurations of STAs connected to AP$_m$ \\
    $\mathcal{C}(d)$ & The set of child nodes of node $d$  \\
    $\hat{\mu}_{d}/{\mu}_{d}$ & The empirical/mean reward of node $d$  \\
    $\hat{\mu}_{d,d_c}/{\mu}_{d,d_c}$ & The empirical/mean reward of node $d$'s child node $d_c$  \\
    $\boldsymbol{w}_{d}^{*}$/$\boldsymbol{\mu}_{d}$ & The optimal allocation/mean reward vector of $\mathcal{C}(d)$\\
    $\Delta_d$ & The suboptimal gap of arm $d$ in the MAB problem \\
    $B_d$ & The empirical best leader among $\mathcal{C}(d)$ \\
    $O_d$ & The transportation cost challenger among $\mathcal{C}(d)$ \\
    $\eta$ & The indicator for the converged layer of the MCT \\
    $\mathcal{A}_{\epsilon}(\boldsymbol{\mu})$ & The $\epsilon$-optimal arms set in terms of $\mathcal{A}(\boldsymbol{\mu})$ with $\boldsymbol{\mu}$\\
    $q^h$ & The $q$-th node at layer $h$ \\
    $T_{\epsilon}^{h}(\boldsymbol{\mu}_{q^h})$ & The asymptotic characteristic time of the $q$-th node at layer $h$ with the stopping rule $\text{GLR}_\epsilon$\\ 
    $T_{\epsilon,\beta}^{h}({\boldsymbol{\mu}_{q^h}},d)$ & The $\beta$-characteristic time of the $q$-th node at layer $h$ with the stopping rule $\text{GLR}_\epsilon$ of arm $d$ \\ 
    $\tau_{\epsilon,\delta}^{h}(q)$ & The convergence time (of the $q$-th node) at layer $h$ with the stopping rule $\text{GLR}_\epsilon$ \\
    $\mathbb{N}_{t,d}$ & The number of times node $d$ is selected at time slot $t$ \\
    $d_{t,h}^*$ & The best arm at layer $h$ following the converged nodes at time slot $t$ \\
    $C_{\boldsymbol{\mu}}$ & The number of equivalence classes, $C_{\boldsymbol{\mu}}:=|\{\mu_d\mid d\in \mathcal{A}(\boldsymbol{\mu})\}|$ \\
    $\mathscr{C}_{\boldsymbol{\mu}}(d)$ & Equivalence class, $\mathscr{C}_{\boldsymbol{\mu}}(d):=\{d\in \mathcal{A}(\boldsymbol{\mu})\mid\mu_{d^*}-\mu_d=\Delta_d\}$ \\
    $\hat{\mathcal{I}}_t$ & The arm with the largest empirical reward at time slot $t$ \\
    $d_{t,h}$ & The node selected at layer $h$ at time slot $t$\\
    \hline  
    \noalign{\hrule height 0.6pt} 
  \end{tabular}  
\end{table}

MAB frameworks provide efficient mechanisms for addressing sequential decision-making problems through balanced exploration-exploitation strategies \cite{bubeck2012regret}. These frameworks find extensive application in reinforcement learning problems \cite{liu2024combinatorial} and wireless network optimization \cite{guo2024fair}, effectively managing uncertainty and dynamic environmental conditions. However, sample complexity remains a persistent challenge in MAB problems, particularly when confronting large arm spaces. To mitigate this challenge, existing approaches either decompose problems into multiple sub-problems \cite{bubeck2012regret} or leverage domain knowledge to accelerate exploration-exploitation processes \cite{combes2018optimal, poiani2024best}. However, these solutions typically cater to specific scenarios or problem domains. This paper leverages LLMs to provide high-quality initializations for MAB algorithms, thereby improving both sample efficiency and generalizability.

The emergence of LLMs represents a paradigmatic shift in wireless network applications, catalyzing breakthrough innovations across diverse domains Ref. \cite{chen2023introspective, zhang2024solving, zhou2024large}. Ref. \cite{chen2023introspective} proposed an innovative prompting paradigm that induces self-optimization concepts to enhance LLM performance in few-shot and zero-shot learning scenarios without requiring fine-tuning. Ref. \cite{zhang2024solving} introduced the OptLLM framework, which synergistically integrates LLMs with external solvers to refine optimization problem modeling and solution processes through natural language queries. In addition, \cite{zhou2024large} explored LLMs' potential in wireless network optimization through ICL, demonstrating competitive performance with traditional deep reinforcement learning approaches in base station power control tasks. This paper explores the capabilities of LLMs in solving MAB problems with large arm spaces, highlighting their effectiveness in facilitating efficient knowledge transfer between small- and large-scale scenarios by bridging optimization-based and data-driven methodologies.

\section{System Model}\label{SecSM}
We consider an uplink WiFi 7 network, as illustrated in Fig.~\ref{scenario}, consisting of a central coordinator and $M$ BSSs.  
Each BSS contains one AP-MLD and several STA-MLDs\footnote{In the following, the terms AP-MLD and AP, as well as STA-MLD and STA, are used interchangeably.}.  
Let $\mathcal{M} = \{1, \ldots, M\}$ and $\mathcal{N} = \{1, \ldots, N\}$ denote the sets of APs and STAs in this network, respectively.  
The set of STAs $\mathcal{N}_m$ is associated with the $m$-th AP (or BSS), where $|\mathcal{N}_m| = N_m$ and $N = \sum_{m=1}^M N_m$.  
In MLO, each STA can support up to $l$ simultaneous transmission links across the $\{2.4, 5, 6\}$ GHz frequency bands~\cite{deng2020ieee}.  
Let $L(n, m)$ denote the set of links where the $n$-th STA transmits to the $m$-th AP, and $|L(n, m)| \leq l$.  
It is worth noting that each frequency band can support multiple channels.

We consider a saturated scenario where each STA always has packets ready for transmission.  
The CSMA protocol is employed to coordinate transmissions among STAs operating on the same channel via carrier-sensing and back-off mechanisms.  
In dense WiFi 7 networks, as shown in Fig.~\ref{scenario}, a central coordinator is deployed to facilitate information exchange among multiple APs.  
The objective of the system is to maximize network throughput by deploying an intelligent channel allocation algorithm at the central coordinator.

\subsection{Multi-Link Operation}
MLO allows devices to establish multiple simultaneous connections across different bands and channels, featuring two operation modes in WiFi 7: Simultaneous Transmit and Receive (STR) and Non-STR (NSTR). In NSTR, cross-link interference prevents a single MLD from performing simultaneous transmission and reception. The MLD will assess whether to activate multiple links. If it does, the MLD evaluates the backoff states at each link to ensure temporal alignment from start to finish, synchronizing transmissions across the multiple links; Otherwise, the MLD transmits data at the first competing link, and during this time, it cannot sense the states of channels of other links. In contrast, STR enables independent link operations with asynchronous data transmission and reception, maintaining distinct channel access parameters for each link. Therefore, the STR mode is more flexible and efficient than the NSTR.

Fig.~\ref{MLO} illustrates the transmission modes of SLO, NSTR, and STR. 
In SLO, the transmitter can use only one link to transmit data at any given time.
In NSTR, multiple links connected to an MLD can either synchronize and transmit simultaneously or operate with only a single active link. For instance, the MLD defers Link 2's back-off counter to allow the simultaneous transmission of Packets 1 and 2 on both links. However, when transmitting Packet 3, the exclusive channel occupancy of Link 1 inhibits the carrier-sensing capability of Link 2, thereby suspending the back-off process on Link 2.
In STR,  packets can be transmitted independently on different links. As a result, Packets 1–4 can be transmitted asynchronously and concurrently on Links 1 and 2, significantly reducing the overall transmission time.
In this paper, we investigate the dynamic channel allocation problem under the STR mode.

\begin{figure}[!t]
\centerline{\includegraphics[scale=0.3]{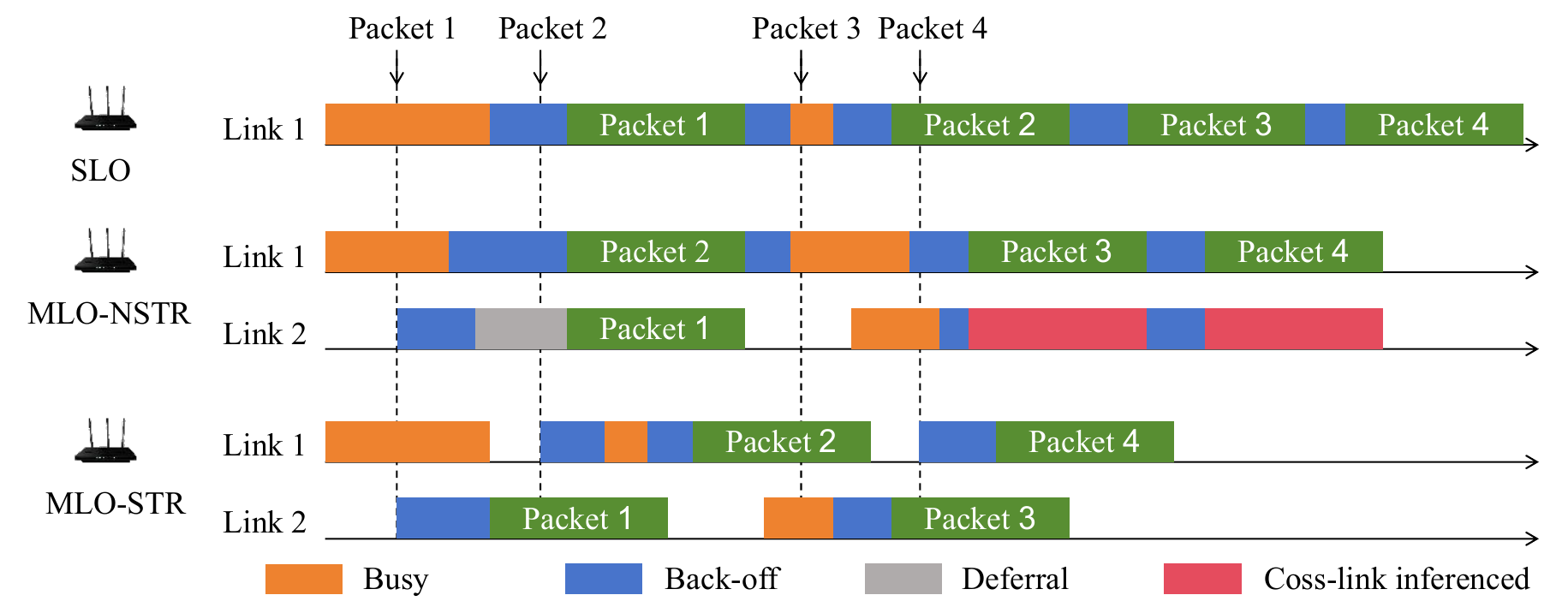}}
\caption{An illustration of the SLO, NSTR, and STR transmission modes in the IEEE 802.11be standard.}
\label{MLO}
\end{figure}

\subsection{Signal Model}\label{Sec_SigMod}
For uplink transmission, the signal received at link $k$ of the $m$-th AP from the $n$-th STA with transmit power $P^k_{n,m}$ is expressed as 
\begin{equation}\label{SigMod}
y_{n,m}^{k}= \sqrt{P^k_{n,m}}\tilde{h}^k_{n,m} \hat{h}^k_{n,m} x^k_{n,m} + z_{n,m}^{k} + w_{m},
\end{equation}
where $x_{n,m}^{k}$ is the normalized signal and $w_{m}$ is the background noise with the normalized power $\sigma_{m}^{2}$. Terms $\tilde{h}^k_{n,m}$ and $\hat{h}^k_{n,m}$ are the large- and small-scale channel fading between the $n$-th STA and the $m$-th AP at link $k$.
In addition, $z_{n,m}^{k}$ is the interference from other links. Notice that Eq.~\eqref{SigMod} represents the components of the received signals, which can be used to calculate the Signal-to-Interference-plus-Noise Ratio (SINR). Let $I_{n,m}^{k}$ be the power of $z_{n,m}^{k}$, which can be calculated by
\begin{equation}\label{Inter}
I_{n,m}^{k}={\sum_{u \in \mathcal{M}} \sum_{v \neq n, v \in \mathcal{N}_u}\sum_{j\in\Lambda_{k}} |\tilde{h}^j_{v,m} \hat{h}^j_{v,m}|^2}P_{v,u}^{j},
\end{equation}
where $m, u\in \mathcal{M}$ are the index of APs and $n \in \mathcal{N}_m,v \in \mathcal{N}_u$ are the index of STAs. 
In addition, $\Lambda_{k}$ is the set of links activated simultaneously on the same channel of link $k$ and $P_{v,u}^{j}$ denotes the transmit power of the link $j$ from the $v$-th STA to the $u$-th AP. Terms $\tilde{h}^j_{v,m}$ and $\hat{h}^j_{v,m}$ are the large- and small-scale channel fading between the $v$-th STA and the $m$-th AP. 
Thus, the SINR at link $k$ of the $m$-th AP is given by
\begin{equation}\label{sinr}
\Gamma_{n,m}^{k}=\frac{|\tilde{h}^k_{n,m} \hat{h}^k_{n,m}|^2 P_{n,m}^{k}}{I_{n,m}^{k}+\sigma^2_{m}}.
\end{equation}

In practice, each transmission can only support a limited number of transmit rates, depending on the number of modulation and coding schemes (MCSs) \cite{tong2021throughput}. Typically, the transmit rate is determined by looking up the table based on the received instantaneous SINR. Mathematically, let $\mathfrak{C}_{n,m}^{k}$ be the set of transmit rates at link $k$ of the $n$-th STA and the $m$-th AP. The STAs can adjust the transmit rates by mapping the instantaneous SINR to a suitable transmit rate. This process can be expressed as $f\left(\Gamma_{n,m}^{k}\right) \rightarrow c^k_{n,m} \in  \mathfrak{C}_{n,m}^{k} $, where $f(\cdot)$ is a mapping function and $c_{n,m}^{k}$ is the obtained transmit rate. Note that a higher SINR corresponds to a higher transmit rate but with a lower success probability. 

\subsection{Carrier-Sensing Graph}\label{Sec_CSG}
Under the CSMA protocol, each transmitter performs the carrier-sensing mechanism before transmitting data on the allocated channel. Consider that link $j$ at the $v$-th STA and link $k$ at the $n$-th STA are assigned to the same channel. They cannot be activated simultaneously if one of their received power exceeds the carrier-sensing threshold. 
For example, the $v$-th STA cannot transmit on link $j$ if 
\begin{equation}\label{Sensing}
|\tilde{h}^j_{n,v} \hat{h}^j_{n,v}|^2 P_{n,m}^{k}\geq S_v,
\end{equation}
where $S_v$ is the carrier-sensing threshold at the $v$-th STA and $P_{n,m}^{k}$ is the transmit power of link $k$ from the $n$-th STA to the $m$-th AP. Terms $\tilde{h}^j_{n,v}$ and $\hat{h}^j_{n,v}$ are the large- and small-scale channel fading between the $n$-th STA and the $v$-th STA. To capture these relationships among multiple links, we construct an undirected carrier-sensing graph $\mathcal{G}(\mathcal{V},\mathcal{E})$ for the WiFi 7 network. Each node in $\mathcal{V}$ represents a link assigned on the same channel, and an edge in $\mathcal{E}$ exists between two nodes if they can't be active simultaneously.

Based on the carrier-sensing graph $\mathcal{G}(\mathcal{V},\mathcal{E})$, we can determine all feasible transmission states in the WiFi 7 network. Defining the set of all feasible states as $\mathcal{F}$, each term $F$ in $\mathcal{F}$ is a vector of length equal to the number of total nodes. When a node is active, the corresponding element in the vector $F$ is $1$, otherwise $0$. Note that only the nodes that do not share a common edge can be active simultaneously. For example, as shown in Fig.~\ref{fig1:a}, when link 3 is active,  the corresponding feasible states are $\mathcal{F} = \{(00100),(01100),(10100)\}$. In other words, link 3 can only be activated simultaneously with link 1 or link 2.

\subsection{Throughput Calculation}
We adopt the Ideal CSMA Network (ICN) model in \cite{jiang2009distributed} to calculate the network throughput using the CTRM model. Fig.~\ref{fig1:b} illustrates the state transition diagram constructed from the carrier-sensing graph of Fig.~\ref{fig1:a} when link 1 is active. As shown in Fig.~\ref{fig1:b}, the state transition of neighboring states satisfies a continuous-time Markov chain.
Let $\mathfrak{s}^k$ and $\mathfrak{b}^k$ be the transmission time and back-off time of the transmitter of link $k$, respectively, which are random variables that may follow an arbitrary distribution \cite{jiang2009distributed}. In addition, $\mathbb{E}[\mathfrak{b}^k]$ and $\mathbb{E}[\mathfrak{s}^k]$ are the means of variables $\mathfrak{s}^k$ and $\mathfrak{b}^k$, respectively.
In fact, Fig.~\ref{fig1:b} constitutes a Markov chain, where the left-hand state transfers to the right-hand state with a probability of ${1}/{\mathbb{E}[\mathfrak{b}^k]}$, and the right-hand state transfers to the left-hand state with a probability of ${1}/{\mathbb{E}[\mathfrak{s}^k]}$. We define $\rho_k$ as the access intensity at link $k$, which is the ratio of the mean transmission time and the mean back-off time at link $k$, i.e., $\rho_k={\mathbb{E}[\mathfrak{s}^k]}/{\mathbb{E}[\mathfrak{b}^k]}$.
\begin{figure}
\centering
\captionsetup[subfloat]{labelsep=space, font=scriptsize}  
\subfloat[Carrier-sensing graph]{\label{fig1:a}
\includegraphics[scale=0.36]{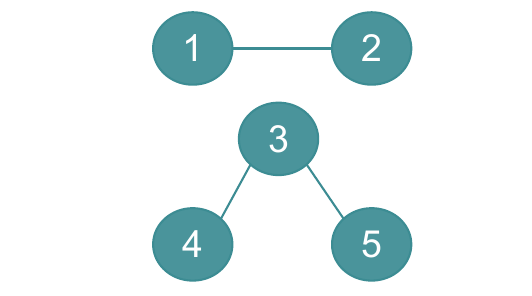}}   
\hfil
\subfloat[State transition diagram]{\label{fig1:b} 
\includegraphics[scale=0.38]{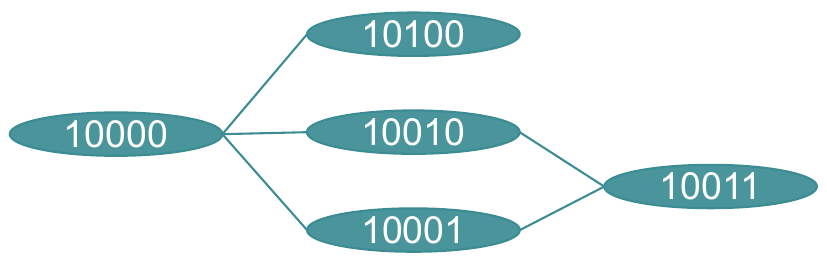}}   
    \caption{ (a) The carrier-sensing graph of five links on the same channel;
(b) The state transition diagram of five links when link 1 is active.}   
\label{fig1}   
\end{figure}

According to \cite{jiang2009distributed}, the stationary state probability of feasible state $F_i \in \mathcal{F}$ is calculated by
\begin{equation}\label{Thro}
P_{{F}_i}=\frac{\prod_{{F}_i(k)\in {F}_i} \rho_{k}^{{F}_i(k)}}{\sum_{{F}_i\in \mathcal{F}}\prod_{{F}_i(k)\in {F}_i} {\rho_{k}^{{F}_i(k)}}}, 
\end{equation}
where $F_i (k)$ is the $k$-th element of the $i$-th vector in the feasible states set $\mathcal{F}$, representing the active (`1') or passive (`0') state of the $k$-th link in this feasible state. Therefore, the throughput of the $n$-th STA and the $m$-th AP at link $k$ can be obtained by summing over the product of all stationary state probabilities $P_{{F}_i}$ and their corresponding transmission rates $c_{n,m}^{k}(F_i)$ of all feasible states, in which the $k$-th element is active. Then, the throughput $\mathcal{T}(k,n,m)$ is calculated by
\begin{equation}
\mathcal{T}(k,n,m)=\sum_{{F}_{i}:{F}_{i}(k)=1}c_{n,m}^{k}(F_i)P_{{F}_i}.
\end{equation}

 For the STR mode, the throughput of different links on an MLD can be considered independent as they transmit on different channels independently. 
 Hence, the total throughput of an MLD can be obtained by accumulating the throughput of each connected link, which can be expressed as
\begin{equation}
\mathcal{T}(n,m)=\sum_{k \in L{(n,m)}}\mathcal{T}(k,n,m),
\end{equation}
where $L{(n,m)}$ is the set of links assigned to the ${n}$-th STA  by the ${m}$-th AP under MLO.

\section{Problem Formulation}\label{SecPF}
This section studies the channel allocation problem in dense WiFi 7 networks with MLO under the STR mode. We first formulate this problem as a combinatorial optimization problem in Section~\ref{SecPF-A}. To overcome the lack of prior information on the network, we further model this problem as an MAB problem in Section~\ref{SecPF-B} to learn the best channel allocation strategy online. Finally, we provide some intuitions and rationalities behind the solutions of the MAB problem. 

\subsection{The Combinatorial Optimization Problem}\label{SecPF-A}
The system's goal is to maximize the overall throughput\footnote{Notice that the objective function can be generalized to any network utility function, such as proportional and max-min fairness.} of the WiFi 7 network by allocating the channels to different STA-MLDs. This problem can be formulated as a combinatorial optimization problem, i.e.,
\begin{equation}\label{total}
\begin{aligned}
& \underset{}{\max\limits_{L_t{(n,m)}}}
& &  \sum^T_{t = 1} \sum_{m \in \mathcal{M}} \sum_{n  \in \mathcal{N}_m} \sum_{k \in L_t{(n,m)}} \mathcal{T}_t(k, n, m) \\
& \mathrm{s.t.}
& & |L_t{(n,m)}| \leq l, \ \forall n, m, t,\\
\end{aligned}
\end{equation}
where $t=1,\ldots, T$ is the time slot required to update the channel allocation strategy. Notice that the time slot duration is longer than the CSMA back-off interval to ensure the Markov chain stationarity for valid throughput calculation. In addition, the constraint $|L_t{(n,m)}| \leq l$ means that each STA-MLD can support up to $l$ simultaneous transmission links. Typically, $l$ is set to 3 in the STR mode, i.e., STA-MLDs can access up to one link within each of the $\{2.4, 5, 6\}$ GHz frequency bands. 

In fact, all MLO configurations can be viewed as discrete feasible solutions. If the network state (or Graph $\mathcal{G}(\mathcal{V},\mathcal{E})$) is known a priori, problem \eqref{total} can be solved using optimization methods. However, directly solving this problem in WiFi 7 networks is challenging. First, constructing all carrier-sensing graphs in a network is computationally intensive, which is an NP-hard problem in dense networks \cite{jiang2009distributed}. In addition, the number of feasible solutions is extremely large due to the numerous combinations of frequency bands and channels in WiFi 7. To overcome these challenges, we resort to online learning theory to obtain the best channel allocation strategy by formulating this problem as an MAB problem.

\subsection{The Multi-Armed Bandit Problem}\label{SecPF-B}
MAB is an efficient framework for handling sequential decision-making problems by balancing the exploitation and exploration dilemma in the learning process \cite{bubeck2012regret}. Thus, we formulate the dynamic channel allocation problem as a stochastic MAB problem where the central coordinator is viewed as the agent. The reward is the overall network throughput, which is a random variable due to the dynamic environment. The arm is the combination of bands and channels (i.e., MLO configurations) of all STAs. The agent aims to maximize the network throughput by selecting the best arm quickly.

We define some symbols and notations related to the MAB problem. Let $\mathcal{J}(m)$ denote the set of feasible configurations for the $m$-th AP, defined as the Cartesian product $\prod_{n \in \mathcal{N}_m} \mathcal{L}(n,m)$, where each $\mathcal{L}(n,m)$ represents the set of all feasible solutions of $L(n,m)$. Thus, the arms set $\mathcal{A}$ is defined as $\mathcal{A}=\prod_{m \in \mathcal{M}} \mathcal{J}(m)=\{a^1, a^2, \ldots, a^{|\mathcal{A}|}\}$. In addition, the reward of arm $a^i$ is denoted by the normalized network throughput, i.e., $r^i\doteq {\mathcal{T}(a^i)}/{\max_{j}\mathcal{T}(a^j)}$.
Define the best arm as $a^{*}$, which has the largest reward $r^{*}$. Then, the set of $\epsilon$-optimal arms can be defined as $\mathcal{A}_{\epsilon} \doteq \{ a^{i} | r^{*} - r^{i} \leq \epsilon \}$.

There are several typical algorithms for solving the above stochastic MAB problem, such as the $\epsilon$-greedy algorithm \cite{auer2002finite}, the Upper Confidence Bound (UCB) algorithm \cite{auer2002finite}, and the Thompson sampling algorithm \cite{tong2023model}. However, they all suffer from the complexity issue, a critical issue in wireless networks, especially when the arm space is huge. There are two promising ways to reduce the sample complexity: 1) enhancing the capability to balance the exploitation and exploration dilemma; 2) reducing the number of arms required to explore.
For the former, we introduce the BAI-MCTS algorithm by integrating the BAI and MCTS algorithms to efficiently balance the exploration and exploitation dilemma, as well as guarantee the quality of allocations. For the latter, we present the LLM-BAI-MCTS algorithm that leverages LLMs to provide an effective initial solution for BAI-MCTS. We will show that both algorithms can find a good channel allocation for the WiFi links and efficiently reduce sample complexity. Moreover, the LLM-BAI-MCTS algorithm can enhance the generalizability of the BAI-MCTS algorithm across diverse network scenarios.

\begin{figure}[!t]
\centerline{\includegraphics[scale=0.51]{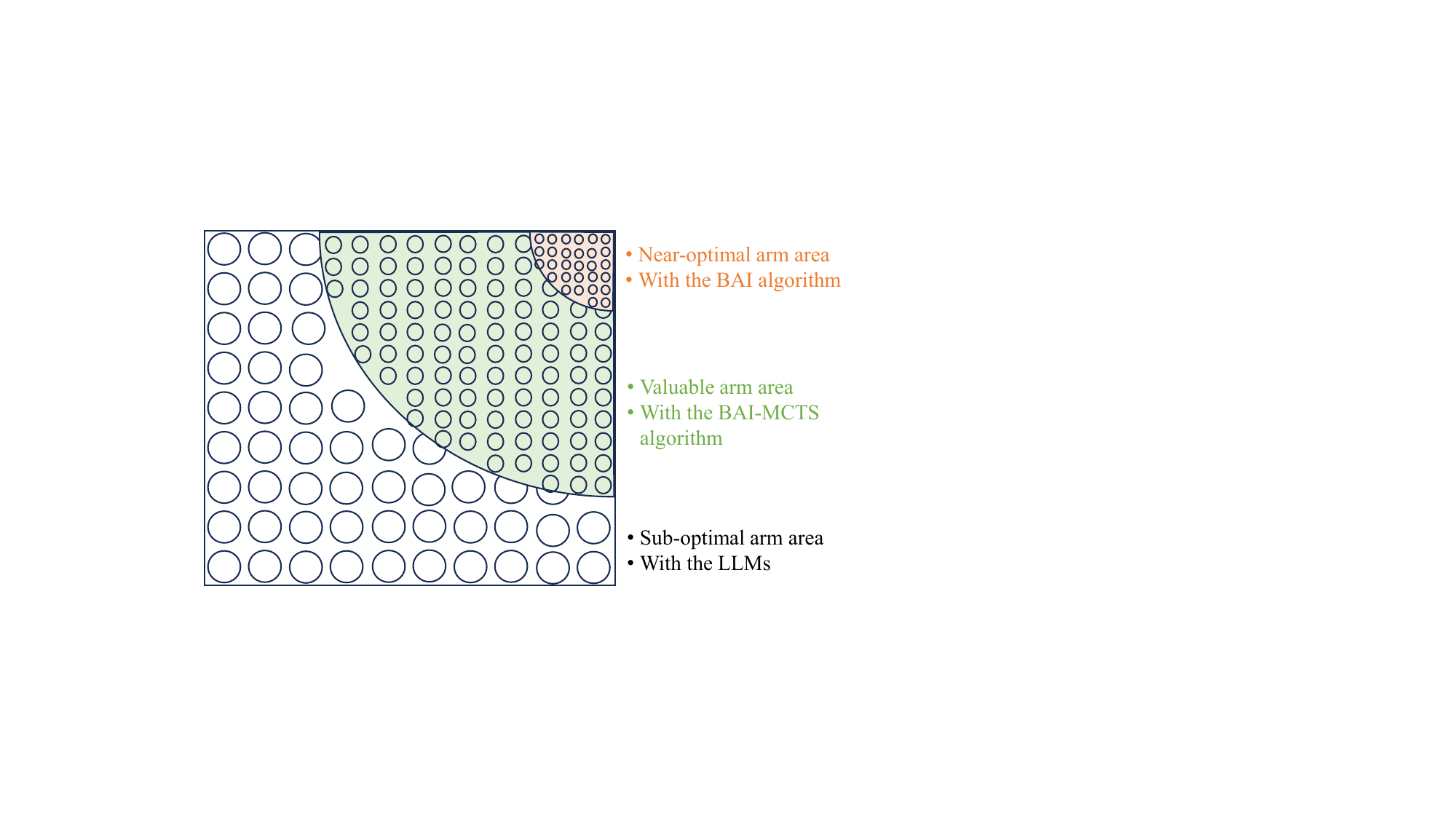}}
\caption{An illustration of the intuitions behind the proposed algorithms.}
\label{Insights}
\end{figure}
The intuitions behind the proposed algorithms are illustrated in Fig.~\ref{Insights}. Each circle represents an arm, with the larger circles being farther from the optimal solution. The optimal and near-optimal arms\footnote{Note that the near-optimal arm is the arm with a value very close to that of the optimal arm.} are located in the top right, while suboptimal arms are scattered in the lower left. As pointed out in \cite{audibert2010best}, in many applications, most of the arms are located in the suboptimal arm area. In contrast, the near-optimal arms are well-centralized around the optimal arm. This indicates that only the arms located in the middle area require more exploration. Therefore, our intuition is to run the algorithm in this valuable arm area to find a sufficiently good solution, ignoring the suboptimal and near-optimal arm areas. To achieve this, we first propose the BAI-MCTS algorithm to find a sufficiently good arm by not obsessively targeting the near-optimal arm area. Thereafter, we utilize LLMs to provide an effective initial solution for the BAI-MCTS algorithm, preventing unnecessary exploration in the sub-optimal arm area. In the following, we introduce the BAI-MCTS and LLM-BAI-MCTS algorithms in Sections~\ref{SecBA} and~\ref{SecLLM}, respectively.

\section{The BAI-MCTS Algorithm}\label{SecBA}
This section presents the BAI-MCTS algorithm to solve the MAB problem. We first introduce the MCTS algorithm to accelerate the arm exploring process in Section~\ref{SecBA-A}, and then present the BAI-MCTS algorithm in Section~\ref{Bai-A} for finding an $\epsilon$-optimal solution for the channel allocation problem. Finally, we derive two upper bounds for the sample complexity and error probability of the BAI-MCTS algorithm in Section~\ref{Theory}.

\subsection{The MCTS Algorithm}\label{SecBA-A}
The MCTS algorithm aims to efficiently explore the decision space and find an optimal strategy with limited resources by balancing the exploration and exploitation dilemma. It first establishes an initial network state at the root node to construct an MCT, which consists of the number of STAs and feasible MLO configurations. The height of an MCT corresponds to the number of STAs. The depth of the root node $d_0$ is set to 0. Specifically, the first-layer nodes represent configurations of set $\mathcal{L}(1,1)$. Then, the second-layer nodes $\mathcal{L}(2,1)$ are expanded from the first-layer nodes, and so on to the final layer. The channel allocations of all STAs are performed by tracing back from the terminal node to the root node. In the MAB problem, we can regard each terminal node as an arm. As shown in Fig.~\ref{MCTS}, the MCTS algorithm contains four steps:
\begin{enumerate}
\item \textbf{Selection strategy} is a key step in MCTS for arm selection. We extend the Empirical-Best and Transportation-Cost (EB-TC$_\epsilon$) algorithm \cite{jourdan2023varepsilon} to select the node of the next layer among the child nodes of the current node. The EB-TC$_\epsilon$ algorithm is explained in detail in Section~\ref{Bai-A}.
\item \textbf{Expansion strategy} is to add a new node to expand the search tree. We perform it by randomly choosing one node from the unselected child nodes of the current node.
\item \textbf{Simulation strategy} is employed to evaluate the selection trajectory to the current node, where an MC simulation is performed to traverse from the expanded node to a terminal node and obtain the corresponding reward.
\item \textbf{Back-propagation strategy} updates the simulation parameters for the selected expansion node and all its ancestor nodes. 
\end{enumerate}
\begin{figure}[!t]
\centerline{\includegraphics[scale=0.28]{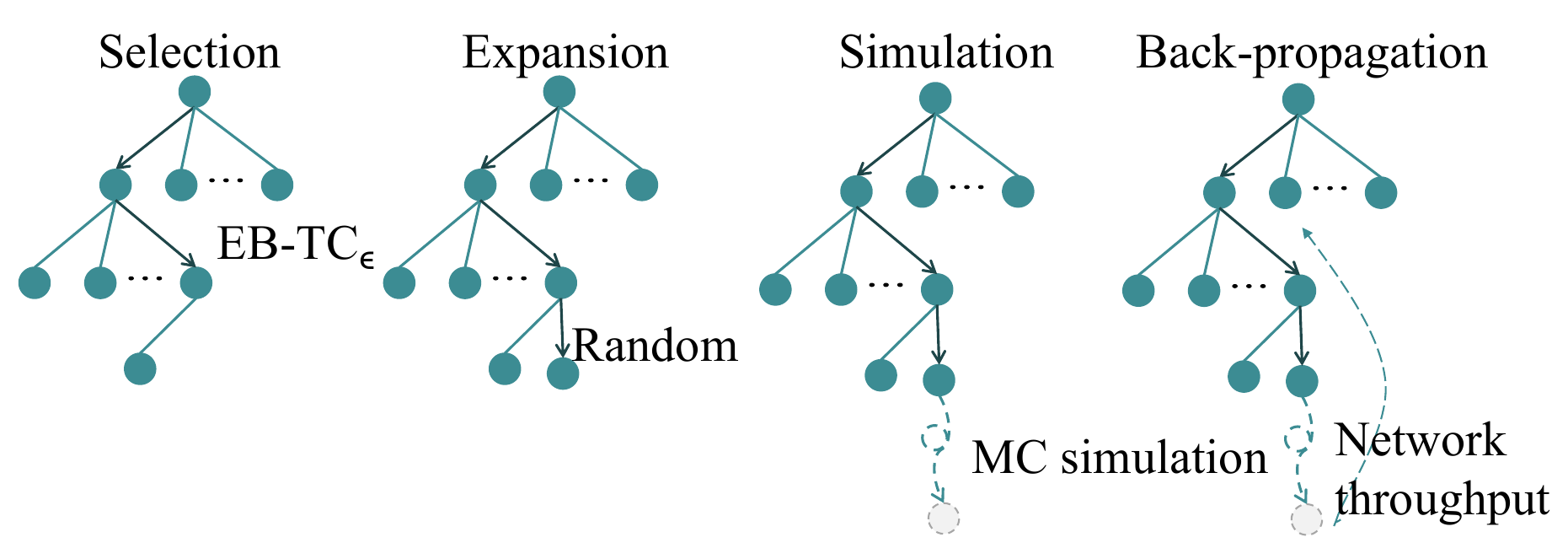}}
\caption{An illustration of four steps of MCTS, containing selection, expansion, simulation, and back-propagation strategies.}
\label{MCTS}
\end{figure}

While the MCTS algorithm provides an efficient way to explore the decision space, the heuristic nature of it cannot guarantee the quality of the solution. To address this issue, we integrate the BAI algorithm into the MCTS algorithm to obtain an $\epsilon$-optimal solution for the MAB problem. 

\subsection{The BAI-MCTS Algorithm}\label{Bai-A}
The pseudocode of the BAI-MCTS algorithm is presented in Algorithm \ref{BAI-MCTS}. Lines 6–15, 18, 19, and 23–30 correspond to the selection, exploration, simulation, and back-propagation steps, respectively. Note that the EB-TC$_\epsilon$ algorithm is an ($\epsilon$, $\delta$)-Probably Approximately Correct (PAC) algorithm.
It satisfies $\mathbb{P}(\tau_{\epsilon,\delta}<+\infty,\hat{\mathcal{I}}_{\tau_{\epsilon,\delta}}\notin\mathcal{A}_{\epsilon})\leq\delta$, where $\tau_{\epsilon,\delta}$ is the convergence time regarding of ($\epsilon$, $\delta$)-PAC, and
$\hat{\mathcal{I}}_{\tau_{\epsilon,\delta}}$ is the arm with the best empirical reward at time slot $\tau_{\epsilon,\delta}$. The key idea of the EB-TC$_\epsilon$ algorithm is to identify the EB leader and the TC$_\epsilon$ challenger, which are the two most promising arms at the current time slot. Next, we show how the EB-TC$_\epsilon$ algorithm can be integrated into the MCTS algorithm. 

First, for a child node $d_c$ at layer $h$ of node $d$, its empirical reward  at time slot $t$ is defined as
\begin{equation}\label{mean}
{\hat{\mu}_{d,d_c} = \frac{1}{\mathbb{N}_{t,d_c}}\sum\limits_{\mathfrak{n} = 1,\ldots,t}r_{\mathfrak{n},d_c}\mathds{1}\left(d_{\mathfrak{n},h}=d_c\right)},
\end{equation}
where $\mathbb{N}_{t,d}$ is the number of times node $d$ is selected at time slot $t$, and $r_{\mathfrak{n},d_{c}}$ is the reward received by the child node $d_c$ at time slot $\mathfrak{n}$. In addition, $d_{\mathfrak{n},h}$ is the node selected at the $h$-th layer of MCT at time slot $\mathfrak{n}$, and $\mathds{1}(\cdot)$ is the indicator function.
If there is no ambiguity, the time slot subscript is omitted in the following.

\begin{algorithm}[!t]
\caption{BAI-MCTS Run by the Central Coordinator}
\label{BAI-MCTS}
\begin{algorithmic}[1]
\State \textbf{Initialize:}  \   $\eta = h = 0$, ${\mathcal{D}}_0 = d_0$, $\bar{\beta}_d = 1/2$, $T_d = \mathbb{N}_{d,d_{c_1}}^{d_{c_2}} = \mathbb{N}_d = 0$, for any node $d$ and $d_{c_1}, d_{c_2} \in \mathcal{C}(d)$
\State \textbf{Input:}  \ $d_0$, $\epsilon$, $\delta$, $N$
\State \textbf{Output:}  \ ${\mathcal{D}}_h$, for ${h = 1, \ldots, N}$
\While {$\eta \neq N$}
	\State $d = {\mathcal{D}}_\eta$
	\While { $\forall~ \mathbb{N}_{\mathcal{C}(d)}\neq 0$ and $h < N$}
		\State $h = \mathcal{H}(d)+1$
		\State choose $B_{d}$, $O_d$ according to (\ref{EB}), (\ref{TC})
		\State $T_{d}(B_d,O_d) = T_{d}(B_d,O_d)+1$
		\State update $\beta_{d}(B_{d},O_d)$, $\bar{\beta}_{d}(B_d,O_d)$ by (\ref{beta}), (\ref{betabar})
			\If { (\ref{choose who}) is satisfied }
				\State update $\mathbb{N}_{d,O_d}^{B_d}$ by (\ref{N}), $d\leftarrow O_d$
            \Else
				\State $d\leftarrow B_d$
\EndIf

\EndWhile		
		\If { $d$ is not the terminal node }
			\State $d\leftarrow$ select one node from $\{ d_c \in \mathcal{C}(d) \,|\, \mathbb{N}_{d_c} = 0 \}$
			\State execute the MC simulation to the terminal node $d_p$
		\Else
			\State $d_p\leftarrow d$
\EndIf  
			\State collect the reward of node $d_p$
			\While {  $\mathcal{H}(d) \neq 0$ }
				\State $\mathbb{N}_d = \mathbb{N}_d+1$, update $\hat{\mu}_{\mathcal{P}(d),d}$ by (\ref{mean})
				\If {{ $\mathcal{H}(d) = \eta+1$ } and { (\ref{stop}) are satisfied }}
						\State$\eta = \eta+1$, ${\mathcal{D}}_{\mathcal{H}(d)} \leftarrow \hat{\mathcal{D}}_{\mathcal{P}(d)}$
\EndIf 
				\State$d\leftarrow \mathcal{P}(d)$		
\EndWhile 
\EndWhile
\end{algorithmic}
\end{algorithm}
Second, the sampling strategy of the EB-TC$_\epsilon$ algorithm is layer-wise, as illustrated in lines 6-15 in Algorithm \ref{BAI-MCTS}. When the agent traverses to node $d$ and each of its child nodes has been visited at least once, the EB leader is obtained as the node with the largest empirical reward in the child nodes, i.e.,
\begin{equation}\label{EB}
B_{d} = \underset{d_c\in \mathcal{C}(d)}{\operatorname*{\arg\max}} \ \hat{\mu}_{d,d_c},
\end{equation}
where $\mathcal{C}(d)$ is the set of child nodes at node $d$. In addition, $\mathcal{P}(d)$ represents the parent node of node $d$. Once the EB leader is identified, the TC$_\epsilon$ challenger can be chosen by 
\begin{equation}\label{TC}
O_d = \underset{d_c\neq B_d, d_c\in \mathcal{C}(d)}{\operatorname*{\arg\min}}\frac{\hat{\mu}_{d, B_d}-\hat{\mu}_{d,d_c}+\epsilon/N}{\sqrt{1/\mathbb{N}_{B_d}+1/\mathbb{N}_{d_c}}}.
\end{equation}
The target average proportion for the leader is defined as
\begin{equation}\label{betabar}
\bar{\beta}_{d}(B_d,O_d) = \frac{(T_d(B_d,O_d)-1)\bar{\beta}_d(B_d,O_d)+\beta_d(B_d,O_d)}{T_{d}(B_d,O_d)},
\end{equation}
where $T_{d}(B_d,O_d)$ is the number of selected times of nodes $(B_d,O_d)$, and
\begin{equation}\label{beta}
\beta_{d}(B_{d},O_d) = \frac{\mathbb{N}_{O_d}}{\mathbb{N}_{B_{d}}+\mathbb{N}_{O_d}}
\end{equation}
is an adaptive proportion.
Let $\mathbb{N}_{d,O_d}^{B_d}$ be the number of times node $O_d$ is selected, where $B_d$ is the leader. If 
\begin{equation}\label{choose who}
\mathbb{N}_{d,O_d}^{B_d}\leq(1-\bar{\beta}_{d}(B_d,O_d))T_{d}(B_d,O_d),
\end{equation} 
the agent selects node $O_d$ and updates
\begin{equation}\label{N}
\mathbb{N}_{d,O_d}^{B_d} = \mathbb{N}_{d,O_d}^{B_d}+1;
\end{equation}
Otherwise, it chooses node $B_{d}$. Therefore, we can fairly explore the two most potential nodes.

In addition to the sampling strategy, a stopping rule is required for EB-TC$_\epsilon$ to function as a fixed-confidence BAI algorithm. The stopping rule adopted in EB-TC$_\epsilon$ is the Generalized Likelihood Ratio ($\text{GLR}_\epsilon$), which is shown in Algorithm~\ref{BAI-MCTS} (lines 26-28).
Let $\eta(\cdot)$ be the stopping indicator, where $\eta = h$ indicates the convergence has occurred up to layer $h$. 
The converged node at layer $h$ is recorded as $\mathcal{D}_h$. Taking the root node $d_0$ as an example, if it satisfies
 \begin{equation}\label{stop}
 \begin{aligned}
 \min_{d_c\neq{\hat{\mathcal{D}}}_{d_0}, d_c \in \mathcal{C}(d_0)}\frac{\hat{\mu}_{d_0,{\hat{\mathcal{D}}}_{d_0}}-\hat{\mu}_{d_0,d_c}+\epsilon/N}{\sqrt{1/\mathbb{N}_{{\hat{\mathcal{D}}}_{d_0}}+1/\mathbb{N}_{d_c}}}\geq\sqrt{2g({\mathbb{N}_{d_0}},\delta)},
 \end{aligned}
 \end{equation} 
the node at the first layer converges to ${\mathcal{D}}_1 = \hat{\mathcal{D}}_{d_0}$. Here, $\hat{\mathcal{D}}_{d_0}$ is the node with the largest current empirical reward among the child nodes of $d_0$. 
In addition, $g(\mathbb{N}_{d_0},\delta)=2\mathcal{Y}(\log((|\mathcal{C}(d_0)|-1)/\delta)/2)+4\log(4+\log(\mathbb{N}_{d_0}/2))$ with $\mathcal{Y}(x) = x+\log(x)$. At the same time, the agent updates the indicator $\eta = \mathcal{H}({\mathcal{D}}_1) = 1$, where $\mathcal{H}(\cdot)$ is the function that returns the depth of the node in the MCT. After the first layer converges, the subsequent sampling process jumps into node ${\mathcal{D}}_1$ (line 5). When $\eta = N$, all layers converge. The BAI-MCTS algorithm outputs the final MLO configuration for the current network.

Fig.~\ref{example} illustrates the MCT construction process of the BAI-MCTS algorithm under a simplified scenario, where there are two AP-STA pairs, each with only two channels. The process comprises five decision-making phases:
\begin{itemize}
\item Phase 1: The MCT starts with a single root node. The algorithm randomly expands a node from the root node in the first step. It then performs the MC simulation from this expanded node until reaching a terminal node (i.e., all AP-STA pairs have their channel configurations). Finally, the reward from this terminal node is back-propagated to update the algorithm's parameters.
\item Phase 2: The algorithm continues exploring the root's unexplored child nodes, then conducting the same simulation-and-backpropagation process as in Phase 1.
\item Phase 3: Once all child nodes of the root node have been explored, the EB-TC$_\epsilon$ algorithm determines which node at the first layer to choose. Partial exploration of the chosen node's child nodes triggers the random expansion and reaches a terminal node. The reward from the terminal node is subsequently back-propagated.
\item Phase 4: This phase is similar to Phase 3, where the EB-TC$_\epsilon$ algorithm guides the selection among all explored child nodes and then follows the expansion and back-propagation process.
\item Phase 5: The agent selects a node in the first layer using the EB-TC$_\epsilon$ algorithm. Since all child nodes of the selected node in the first layer have already been explored, the agent proceeds the same algorithm to select nodes in the next layer, where it encounters a terminal node. Then, back-propagation is performed. 
\end{itemize}

During these five phases, the algorithm generates a sequence of channel combinations, $\{1,2\}, \{2,1\}, \{1,1\}, \{1,2\}, \{1,2\}$. As demonstrated in this specific illustration, the BAI-MCTS algorithm utilizes a hierarchical modified EB-TC$_\epsilon$ algorithm for MCT node selection and the stop criterion determination. The integration of the BAI and MCTS algorithms enables a heuristic $\epsilon$-optimal channel allocation for the throughput maximization problem. This indicates that the BAI-MCTS algorithm can efficiently reduce the sample complexity while optimizing the network throughput.

\begin{figure}[!t]
\centerline{\includegraphics[scale=0.70]{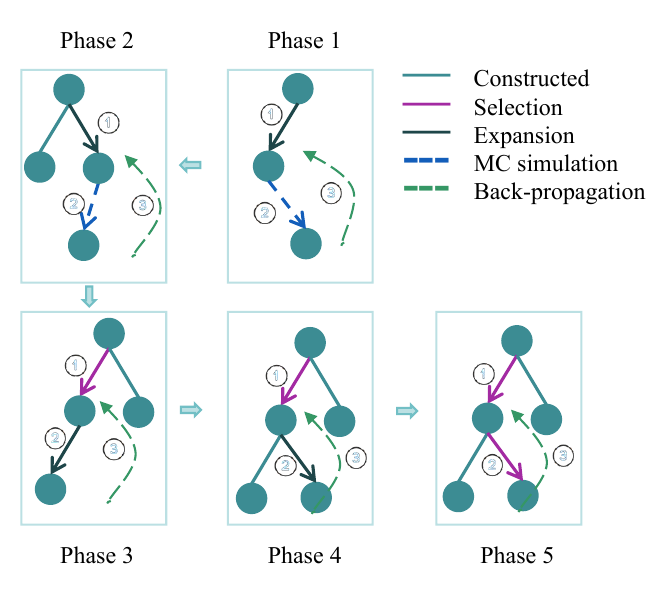}}
\caption{An illustration of the BAI-MCTS algorithm's channel allocation process for two AP-STA pairs, each limited to two MLO configurations.}
\label{example}
\end{figure}

\subsection{Theoretical Analysis}\label{Theory}
Next, we show how the BAI-MCTS algorithm can guarantee an $\epsilon$-optimal solution for the channel allocation problem by deriving two upper bounds for the sampling complexity and the error probability of the proposed algorithm, respectively. For convenience, we first give the notations related to the MCT. Assume that each AP-STA pair has $|{\mathcal{L}}|$ feasible configurations and the network throughput follows a Gaussian distribution $\mathfrak{N}$. In addition, the ${\boldsymbol{\mu}_{q^h}}$ and ${\boldsymbol{w}_{q^h}^*}$ denote the mean reward and the optimal allocation vector of the $q$-th node at layer $h$, respectively. The best node among $\mathcal{C}(q^h)$ is denoted by $d^*({\boldsymbol{\mu}_{q^h}})$, and the suboptimal gap is defined as  $\Delta_{d}:=\{\mu_{{d^*({\boldsymbol{\mu}_{q^h}})}}-\mu_{d}| d\in \mathcal{C}(q^h)\}$. According to \cite{jourdan2023varepsilon}, the asymptotic characteristic time is defined as the time required to verify an $\epsilon$-optimal arm among $\mathcal{C}(q^h)$ with the stopping rule $\text{GLR}_\epsilon$, i.e., 
\begin{equation}
    \begin{split}
        T_{\epsilon}^{h}({\boldsymbol{\mu}_{q^h}})=\operatorname*{min}_{d\in\mathcal{A}_{\epsilon}(\boldsymbol{\mu}_{q^h})}\operatorname*{min}_{\beta\in(0,1)}T_{\epsilon,\beta}^{h}({\boldsymbol{\mu}_{q^h}},d), 
    \end{split}
\end{equation}
where T$_{\epsilon,\beta}^{h}(\boldsymbol{\mu}_{q^h}, d)$ is the $\beta$-characteristic time for arm $d$, which is the time required for identifying whether the arm $d$ is $\epsilon$-optimal or not in terms of proportion $\beta$, i.e., 
\begin{equation}
    \begin{split}
        &{T_{\epsilon,\beta}^{h}({\boldsymbol{\mu}_{q^h}},d)}^{-1} =\\
        &\max _{\boldsymbol{w}_{q^h}\in\triangle_{|{\mathcal{L}}|},{\boldsymbol{w}}_{q^h}(d)=\beta}\min_{d'\neq d}\frac12\frac{({\mu}_{q^h,d}-\mu_{q^h,d'}+\epsilon)^2}{1/\beta+1/\boldsymbol{w}_{q^h}(d')}.
    \end{split}
\end{equation}
The $\triangle_{|{\mathcal{L}}|}$ is a $(|{\mathcal{L}}|-1)$-dimensional probability simplex, i.e., $\triangle_{|{\mathcal{L}}|}:=\left\{\boldsymbol{w}\in\mathbb{R}_{+}^{|{\mathcal{L}}|}\mid  \sum_{i=1}^{|{\mathcal{L}}|}\boldsymbol{w}(i)=1\right\}$. The ${\boldsymbol{w}}_{q^h}(d)$ is the allocation proportion for arm $d$ in terms of ${\boldsymbol{w}}_{q^h}$. We define the characteristic time $\hat T_{\epsilon}^{h}$ at layer $h$ as the maximum characteristic time among all nodes at layer $h$, i.e.,
\begin{equation}
    \begin{split}
        \hat T_{\epsilon}^{h}(\boldsymbol{\mu}^{h})& = \max_{q=1,\ldots,|{\mathcal{L}}|^{h}} T_{\epsilon}^{h}(\boldsymbol{\mu}_{q^h}),
    \end{split}
\end{equation}
where $\boldsymbol{\mu}^{h}$ denotes the mean reward vector for the sets of child nodes associated with all nodes at layer $h$. Now, we can give the upper bound on the expected sample complexity of the BAI-MCTS algorithm in the asymptotic regime.

\begin{theorem}\label{Theorem1}
\textit{For the BAI-MCTS algorithm with the reward bounded within the interval [0, 1] and its distributions vector ${\boldsymbol{\nu}}\sim\mathfrak{N}_{{|\mathcal{L}|}^{N}}$, its upper bound on the expected sample complexity $\mathbb{E}\left[\tau_{\epsilon,\delta}\right]$ in the asymptotic regime is given by
\begin{equation}
    \begin{split}
        \lim\sup_{\delta\to0}\frac{\mathbb{E}_{{\boldsymbol{\nu}}}[\tau_{\epsilon,\delta}]}{\log(1/\delta)}\leq \sum_{h=0}^{N-1}{\hat T_\frac{\epsilon}{N}^{h}}(\boldsymbol{\mu}^{h}).
    \end{split}
    \end{equation}}
\end{theorem}

\begin{proof}
	See Appendix~\ref{convergence}.
	$\hfill\blacksquare$
\end{proof}

\textbf{Remark:} Theorem \ref{Theorem1} establishes that, in the fixed-confidence setting, the proposed algorithm will converge within a finite time slot $\tau_{\epsilon,\delta}$. Moreover, this algorithm can also be used in the anytime setting. The following theorem provides an upper bound on the probability of recommending a non-$\epsilon$-optimal arm when the algorithm has not stopped yet, thereby complementing Theorem \ref{Theorem1}. In practice, due to the relationship $T_{\epsilon,1/2}(\boldsymbol{\mu}) \leq 2T_\epsilon(\boldsymbol{\mu})$ \cite{jourdan2023varepsilon}, we focus on the case where $\beta = 1/2$.

\begin{theorem}  \label{Theorem2}
\textit{At time slot $t$, for ${\boldsymbol{\nu}} \sim \mathfrak{N}_{{|\mathcal{L}|}^{N}}$, $\mathbb{N}_{t,d_{t,\eta}} \geq D_{{\boldsymbol{\mu}}_{d_{t,\eta}}}$, and $\mathbb{N}_{t,d_{t,h}^*} \geq \max\{D_{{\boldsymbol{\mu}}_{d_{t,h}^*}}, T_1({{\boldsymbol{\mu}}_{d_{t,h}^*}})\},\forall h=\eta+1,\ldots, N-1$, if the BAI-MCTS algorithm with a fixed proportion $\beta = 1/2$ has not stopped yet, then $\forall\epsilon\geq0$,}
\begin{equation}
\begin{aligned}
&\mathbb{P}_{\boldsymbol{\nu}}\left(\hat{\mathcal{I}}_t \notin \mathcal{A}_{\epsilon}(\boldsymbol{\mu})\right) \leq \mathds{1}\left(\epsilon < \Delta_{\mathrm{max}}\right)\left(1 - (1-\delta)^{\frac{\eta}{N}} \right) \\
&\times \prod_{h=\eta+1}^{N-1} \frac{1}{4\sqrt{2(|\mathcal{L}| - 1)}} \left(1 - Q\left(|\mathcal{L}|, \mathbb{N}_{t, d_{t, h}^*}, \boldsymbol{\mu}_{d_{t, h}^*}, \frac{\epsilon}{N}\right)\right) \\
&\times \left(1 - Q\left(|\mathcal{L}|, \mathbb{N}_{t, d_{t, \eta}}, \boldsymbol{\mu}_{d_{t, \eta}}, \frac{\epsilon}{N}\right)\right),
\end{aligned}
\end{equation}
\textit{where $T_1(\boldsymbol{\mu})$ is defined in Appendix~\ref{convergence}. Terms $Q(K,t,\boldsymbol{\mu},\epsilon)$ and $D_{\boldsymbol{\mu}}$ are defined in Appendix~\ref{error}.}
\end{theorem}
\begin{proof}
	See Appendix~\ref{error}.
	$\hfill\blacksquare$
\end{proof}

\textbf{Remark:} Theorem~\ref{Theorem1} reveals that the algorithm's sampling complexity has an upper bound closely related to $\epsilon$. By utilizing the BAI technique, the BAI-MCTS algorithm can reduce its reliance on the suboptimal gap, which frequently acts as a bottleneck in existing MAB algorithms. Theorem~\ref{Theorem2} further asserts that, within the non-converged layers, the error probability of BAI-MCTS decreases exponentially as the number of parent-node selections (i.e., the horizon of child nodes) increases.

\section{The LLM-BAI-MCTS Algorithm}\label{SecLLM}
While the BAI-MCTS algorithm provides an efficient search strategy to solve the channel allocation problem, the theoretical analysis reveals that its sampling complexity and error probability depend on the number of layers and nodes. To further reduce the sampling complexity and bolster the generalizability of the proposed algorithm, we propose the LLM-BAI-MCTS algorithm by integrating LLMs into the BAI-MCTS framework. 
The key idea is to provide a high-quality initialization for the BASI-MCTS algorithm by leveraging the LLMs' ICL and CoT capabilities.
This operation enables efficient exploration of the valuable arm area in Fig.~\ref{Insights}, rather than searching the entire space.

\begin{figure}[!t]
\centerline{\includegraphics[scale=0.70]{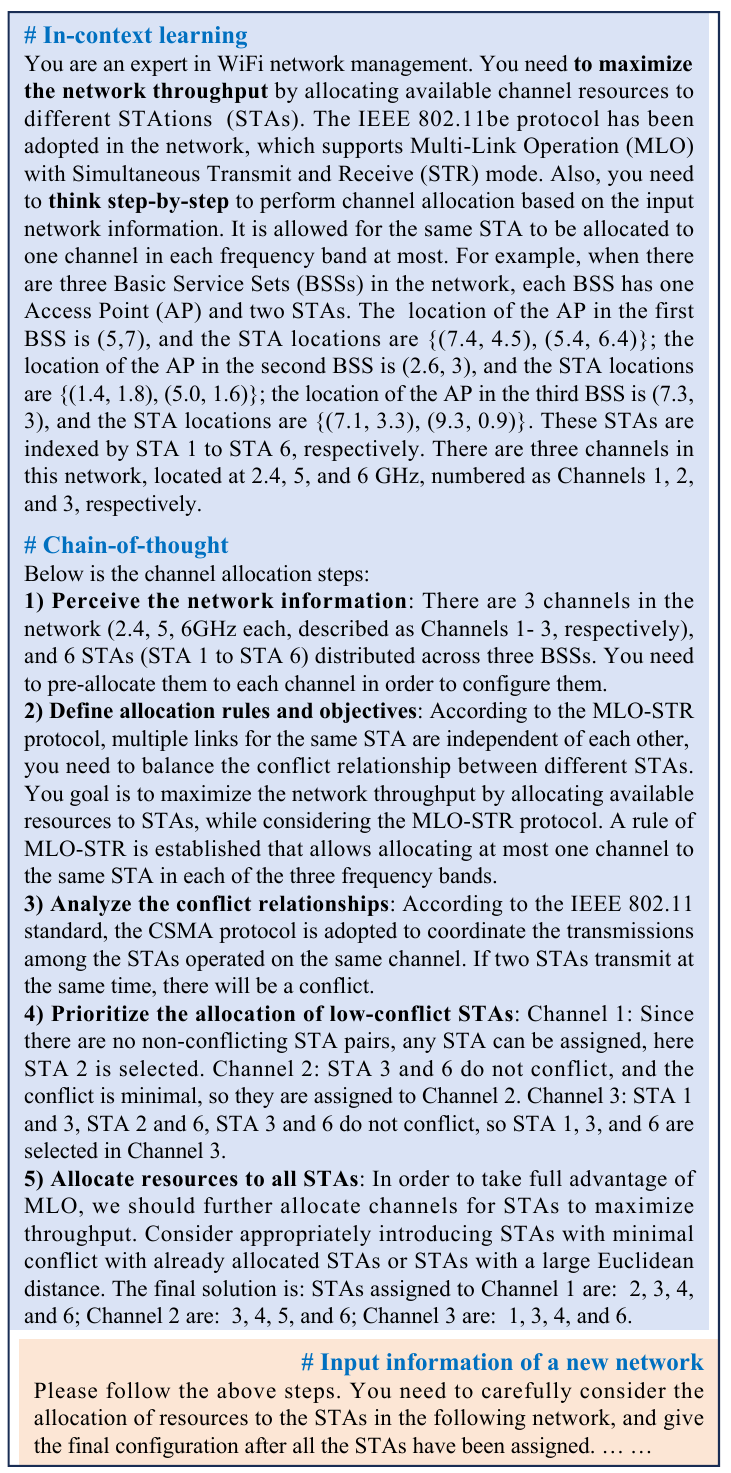}}
\caption{An example of the LLM-assisted channel allocation.}
\label{LLM}
\end{figure}
\subsection{In-Context Learning}
ICL refers to the ability of LLMs to perform specific tasks by leveraging information provided within the input context without turning the model parameters. The model ``learns'' how to respond or perform a task based on the examples and instructions given in the prompt itself. Instead of training the model separately for each task, ICL allows it to adjust its responses based on the context provided at inference time. The model uses these cues to generate relevant responses tailored to the task, which can be expressed as
\begin{equation}
LLM(D_{task}, \mathfrak{E}) \Rightarrow \text{output},
\end{equation}
where $D_{task}$ is the task description, and $\mathfrak{E}$ is the example set.  

In the dynamic channel allocation problem, the task description includes the problem's background, objectives, and network information, as illustrated in Fig.~\ref{LLM}. Specifically, we regard the LLM as a wireless network expert responsible for configuring channels for STAs to maximize overall network throughput. The background knowledge provides the core concepts of WiFi 7 networks, such as CSMA and MLO. The key network characteristics include the number of STAs, STAs' locations, and their conflicting relationships obtained from the carrier-sensing graph. In simulations, we adopt few-shot learning to enable the LLM to understand the network information and extract the structure information of the channel allocation problem. As shown in Fig.~\ref{LLM}, ICL utilizes examples consisting of three BSSs (i.e., 3 APs and 6 STAs) to learn how to perform channel allocation. Notice that the allocation results of these examples are derived by solving problem \eqref{total} with Algorithm~\ref{BAI-MCTS}. The motivation behind this operation is that solutions to combinatorial optimization problems in small-scale scenarios can serve as examples for ICL. This facilitates efficient knowledge transfer between small- and large-scale scenarios by bridging optimization-based and data-driven methods.

\subsection{Chain of Thought}\label{CoT}
CoT is a powerful prompting technique that encourages LLMs to break down their reasoning process into sequential steps, rather than jumping directly to a conclusion.
This approach enhances the model's ability to handle complex tasks by breaking them down into several manageable parts, improving the accuracy and coherence of the responses. 
In the channel allocation problem, we integrate CoT with ICL to find a high-quality initialization for the BAI-MCTS algorithm by inserting several CoT-augmented examples into LLMs. As shown in Fig. \ref{LLM}, this problem is broken into five steps:
\begin{enumerate}
{\item  \textbf{Perceive network information:} This step gathers the details about available channels and the number/distribution of APs and STAs in the network.
\item  \textbf{Define allocation rules and objectives:} It establishes rules and objectives for the task to balance STA conflicts and maximize throughput following MLO.
\item  \textbf{Analyze conflict relationships:} It focuses on the CSMA protocol and STAs’ conflicts obtained by the carrier-sensing graph constructed based on (\ref{Sensing}).
\item  \textbf{Prioritize the allocation of low-conflict STAs:} This step initially assigns channels to STAs with minimal conflicts within each channel. 
\item  \textbf{Allocate resources to all STAs:} It adjusts the channel assignments by leveraging the MLO benefits and enhances throughput by considering additional network factors.}
\end{enumerate}

It can be seen from Fig.~\ref{LLM} that each step is attached with a well-designed prompt and the example explanations for ICL. 
These steps constitute an AI workflow that can handle complex tasks autonomously. 
By providing the problem descriptions and examples to the LLM, it outputs a high-quality initialization for the channel allocation problem. Then, this result is fed into the BAI-MCTS algorithm to find the $\epsilon$-optimal solution for the channel allocation problem.

\begin{algorithm}[!t]
\caption{LLM-BAI-MCTS at the Central Coordinator}
\label{LLM_BAI-MCTS_v1}
\begin{algorithmic}[1]
\State \textbf{Initialize:}  the BAI-MCTS algorithm
\State \textbf{Input:}  Problem descriptions and examples
\State \textbf{Output:} The channel allocations
\State Insert the ICL and CoT prompts into the LLM
\State Input the problem information into the LLM
\State Obtain a high-quality initialization
\State Insert this initialization into the BAI-MCTS algorithm
\State Output channel allocation results
\end{algorithmic}
\end{algorithm}

\subsection{The LLM-BAI-MCTS Algorithm}
Building upon ICL and CoT paradigms, we introduce the LLM-BAI-MCTS algorithm for the channel allocation problem in dense WiFi 7 networks, as shown in Algorithm~\ref{LLM_BAI-MCTS_v1}. 
It integrates LLMs with the BAI-MCTS algorithm to enhance efficiency and performance. The process begins with initializing the BAI-MCTS algorithm, taking problem descriptions and examples as input, and aiming to output optimal channel allocations. The key innovation is in lines 4-7, where the algorithm leverages ICL and CoT prompting techniques to extract domain knowledge from the LLM. By inserting these prompts and providing problem information to the LLM, the algorithm obtains a high-quality initialization that significantly reduces the search space for the BAI-MCTS algorithm. This initialization is then fed into the BAI-MCTS algorithm, which explores the remaining solution space to find the $\epsilon$-optimal channel allocation. 

The high-quality initialization refers to a strategic partitioning of the channel allocation problem. Specifically, the algorithm identifies $\mathfrak{L}$ STAs whose configurations are determined directly by the LLM's output, where $\mathfrak{L}$ represents the subset of STAs benefiting from LLM-assisted allocation. This selection can be prioritized based on either the LLM's confidence levels for particular allocations or specific performance requirements of certain STAs within the network. Meanwhile, for the remaining $N-\mathfrak{L}$ STAs, the algorithm employs the BAI-MCTS approach to determine optimal channel allocations. This hybrid strategy significantly reduces computational complexity by decreasing the height of MCT by $\mathfrak{L}$ layers, as these LLM-determined allocations no longer require exploration within the tree structure. 

The insight of Algorithm~\ref{LLM_BAI-MCTS_v1} lies in leveraging the LLM's capabilities to comprehend complex wireless network dynamics without explicit domain-specific training. When provided with carefully structured background information, the LLM can efficiently contextualize the channel allocation problem. By incorporating representative examples during ICL, the model extracts the underlying mathematical relationships between network topologies and their optimal allocation strategies, enabling generalization to novel configurations. The CoT technology further enhances solution quality by decomposing the complex optimization process into interpretable reasoning steps, promoting both reliability and verifiability of the generated allocations. This combination of domain knowledge and guided reasoning enables the LLM to facilitate efficient knowledge transfer between small- and large-scale scenarios by bridging the optimization-based and data-driven methods.

\section{Numerical Results}\label{SecNA}
We conduct several experiments to evaluate the proposed algorithms with different settings. We consider a scenario of a $10\times10$ m square network, as shown in Fig.~\ref{topology}, where three APs are located in the vertices of an equilateral triangle, and the STAs are uniformly distributed in this square area. In this uplink WiFi 7 network, all the STAs employ the CSMA protocol to coordinate their transmissions. 
\begin{figure}[!t]
\centerline{\includegraphics[scale=0.6]{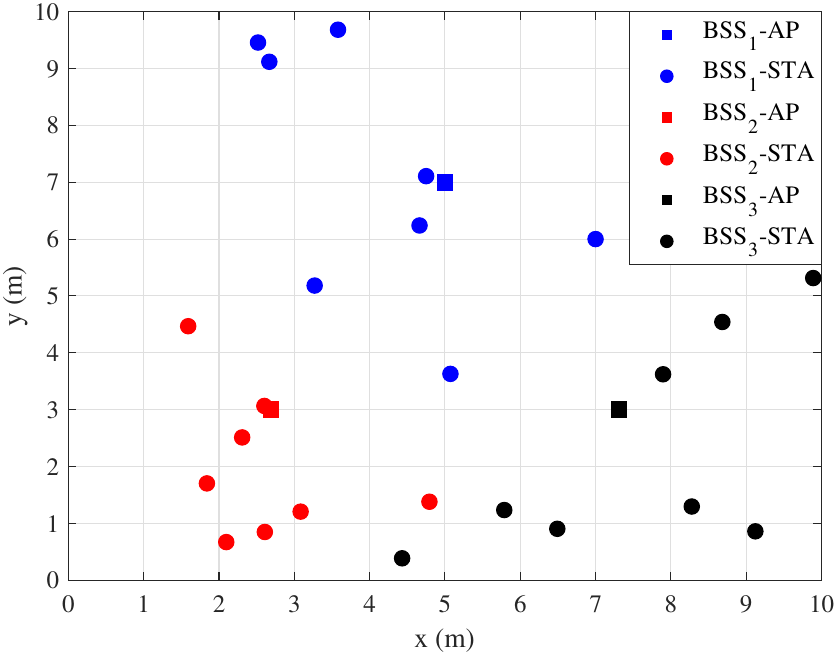}}
\caption{A simple network with three BSSs distributed in a $10\times10$ m area.}
\label{topology}
\end{figure}
\setlength{\textfloatsep}{4pt}
The access intensity \cite{jiang2009distributed} of all links in the ICN is set to $1$. The transmit power is set to $20$ dBm, and the carrier sensing threshold is $-60$ dBm. The normalized power of the background noise is $-95$ dBm. The large-scale channel fading is set to $({3\times 10^8}/{4\pi f_c\mathscr{D}^2(\text{STA}, \text{AP})})^2$, where $f_c=\{2.4, 5, 6\}$ GHz is the carrier frequency and $\mathscr{D}(\text{STA}, \text{AP})$ is the Euclidean distance between a STA and an AP. The small-scale channel fading follows a Rayleigh distribution with a $\sqrt{1/2}$ scale parameter. The available
transmission rates are $\{20, 50, 100, 150\}$ Mbps. All results are obtained from $10^3$ MC experiments.

We adopt a discrete event simulator in \cite{tong2021throughput} to simulate the CSMA protocol and calculate the network throughput. For convenience, we assume that each frequency band only has one channel. Thus, the MLO configurations for each AP-STA pair are $\{(001)$$, (010)$$, (100)$$, (011)$$, (101)$$, (110)$$, (111)\}$, corresponding to arms 1 to 7 in Fig.~\ref{arm}. For example, arm 1 is only assigned to the channel in 2.4 GHz, while arm 7 has three channels (links) in 2.4, 5, and 6 GHz, respectively. In the following, we compare the BAI-MCTS algorithm with several baselines in terms of the network throughput. 
\begin{itemize}
    \item \textbf{The UCB applied to the Tree (UCT) algorithm} selects nodes in MCTS by maximizing the estimated mean rewards using the UCB1 algorithm \cite{auer2002finite};
    \item \textbf{The DGN-MCTS algorithm} \cite{bai2013bayesian} models accumulated reward uncertainty as a Normal distribution mixture and selects nodes via Thompson sampling in MCTS;
    \item \textbf{The random selection method} randomly chooses arms with equal probabilities, disregarding estimated rewards;
    \item \textbf{The optimal value} is obtained by solving problem~\eqref{total} using the exhaustive search method.
\end{itemize}

\begin{figure}[!t]
\centerline{\includegraphics[scale=0.32]{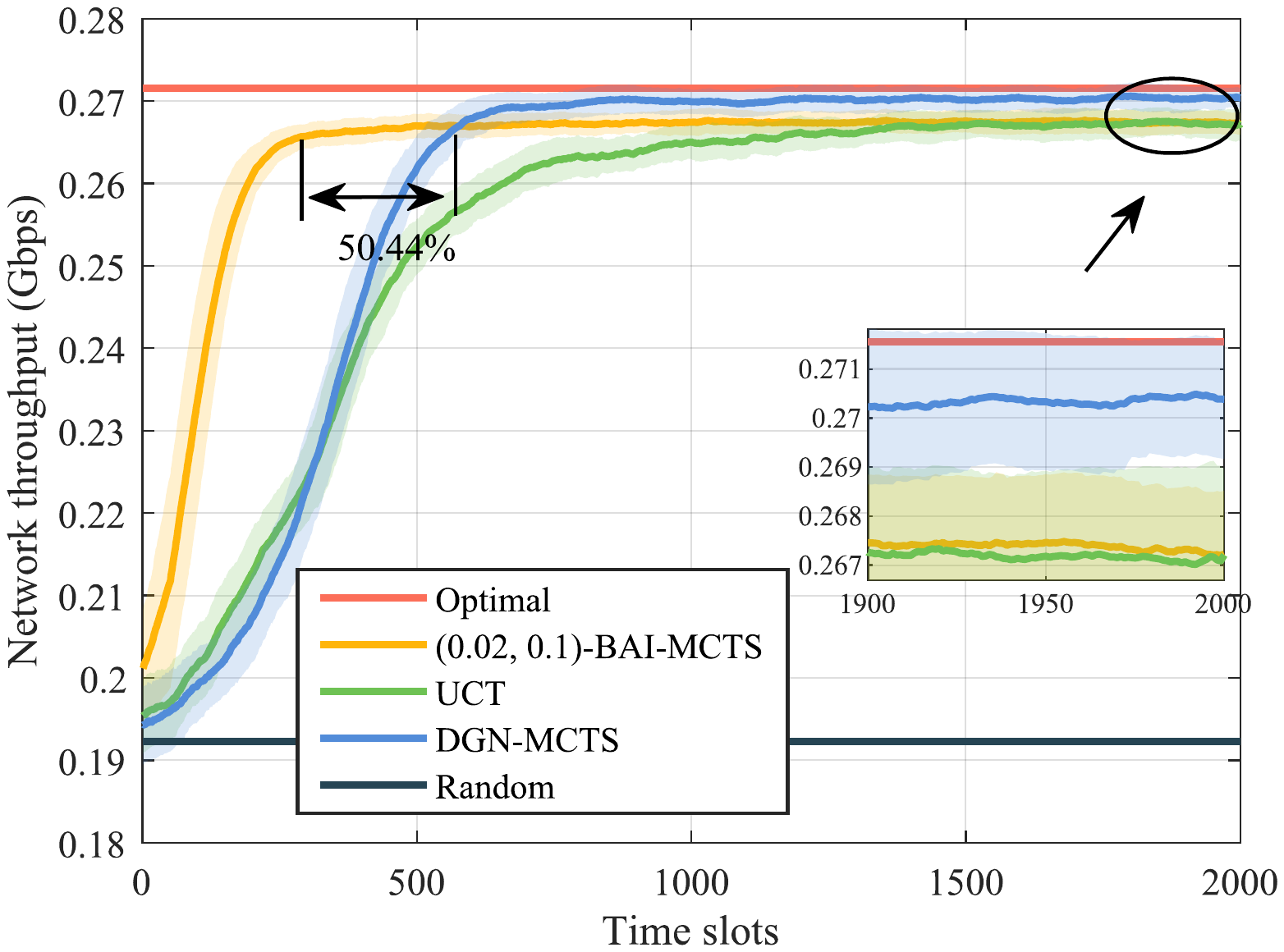}}
\caption{Average network throughput of the optimal value, the BAI-MCTS algorithm, the UCT algorithm, the DNG-MCTS algorithm, and the random selection method, where $T = 2,000$, $\epsilon = 0.02$, and $\delta=0.1$.}
\label{algorithm}
\end{figure}
Fig.~\ref{algorithm} compares the $(0.02, 0.1)$-BAI-MCTS algorithm with the UCT algorithm, the DNG-MCTS algorithm, the random selection method, and the optimal value in the scenario of $3$ APs and $6$ STAs.  We observe that all methods can converge within $1,500$ time slots except for the random selection method. This result is significant as the number of total arms is $7^6 = 117,649$. In addition, the BAI-MCTS algorithm achieves the fastest convergence rate among all algorithms thanks to the BAI technique. Notably, the convergence rate of the BAI-MCTS algorithm is $50.44\%$ faster than the DNG-MCTS algorithm when reaching $98\%$ of the optimal value. Furthermore, the performance gap between the BAI-MCTS algorithm and the optimal value is below $\epsilon$ (i.e., $0.02$), validating the theoretical ($\epsilon$, $\delta$)-optimality guarantee.

\begin{figure}[!t]
\centerline{\includegraphics[scale=0.54]{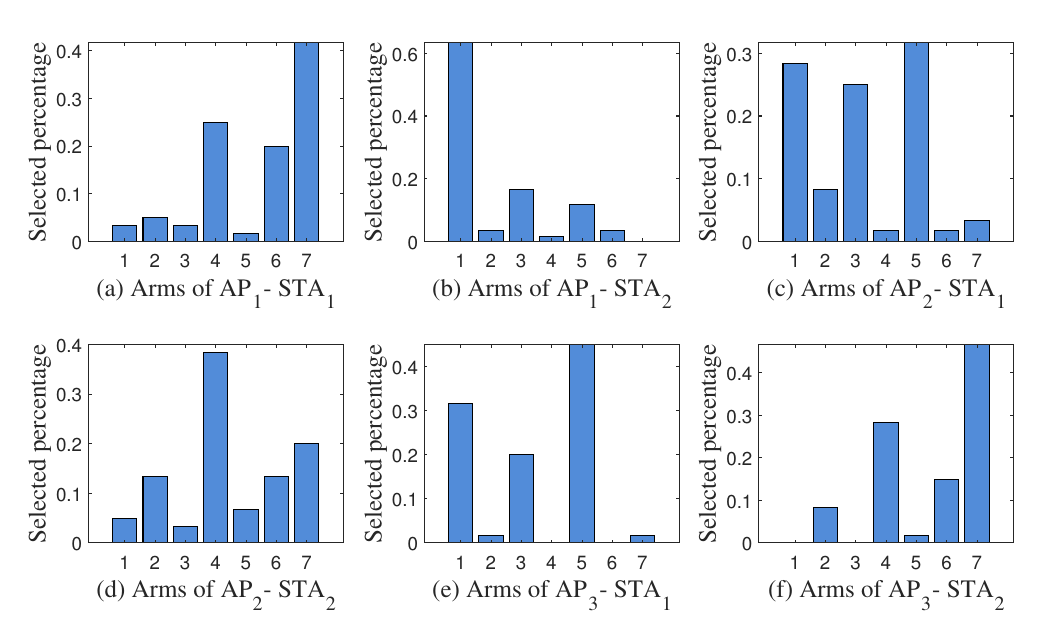}}
\caption{ The selected percentage of each arm of each AP-STA pair in the network by running the BAI-MCTS algorithm.}
\label{arm}
\end{figure}
Fig.~\ref{arm} shows the selected percentage of each arm of each AP-STA pair using the BAI-MCTS algorithm in the scenario of $3$ APs and $6$ STAs. The optimal configuration, determined by the exhaustive search method, is $\{7, 1, 6, 1, 3, 4\}$. From Figs.~\ref{arm}a-\ref{arm}f, we observe that BAI-MCTS tends to select the strategy $\{7, 1, 5, 4, 5, 7\}$, which, while different, closely aligns with the optimal policy. This is because the BAI algorithm aims to identify an $\epsilon$-optimal arm, ensuring a throughput close to the optimal value. This is also observed in Fig.~\ref{algorithm}.

\begin{figure}[!t]
\centerline{\includegraphics[scale=0.63]{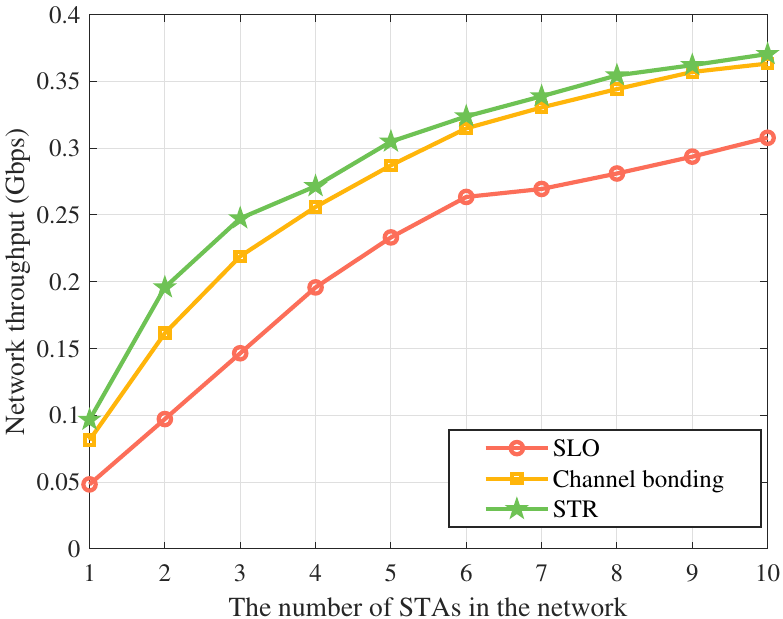}}
\caption{ Network throughput versus the number of STAs in the network of the BAI-MCTS algorithm for $1,000$ random network typologies, where $T = 2,000$, $\epsilon = 0.02$, and $\delta=0.1$.}
\label{mechanism}
\end{figure}

Fig.~\ref{mechanism} shows the performance of the $(0.02, 0.1)$-BAI-MCTS algorithm under varying transmission modes, i.e., SLO, channel bonding, and STR, as the number of STAs increases. The available resources are two 20 MHz channels in the 5 and 6 GHz bands, respectively. The AP-STA pair in SLO mode can select only one channel, while channel bonding can dynamically aggregate two channels. In addition, the STR mode allows simultaneous transmission in both 5 GHz and 6 GHz bands. From Fig.~\ref{mechanism}, we see that STR  achieves the highest network throughput among the three modes, demonstrating its efficiency in utilizing channel resources. As the number of STAs increases, the performance gap between the STR and channel bonding modes narrows due to increased competition, reflecting realistic network behavior.

\begin{figure}[!t]
\centerline{\includegraphics[scale=0.32]{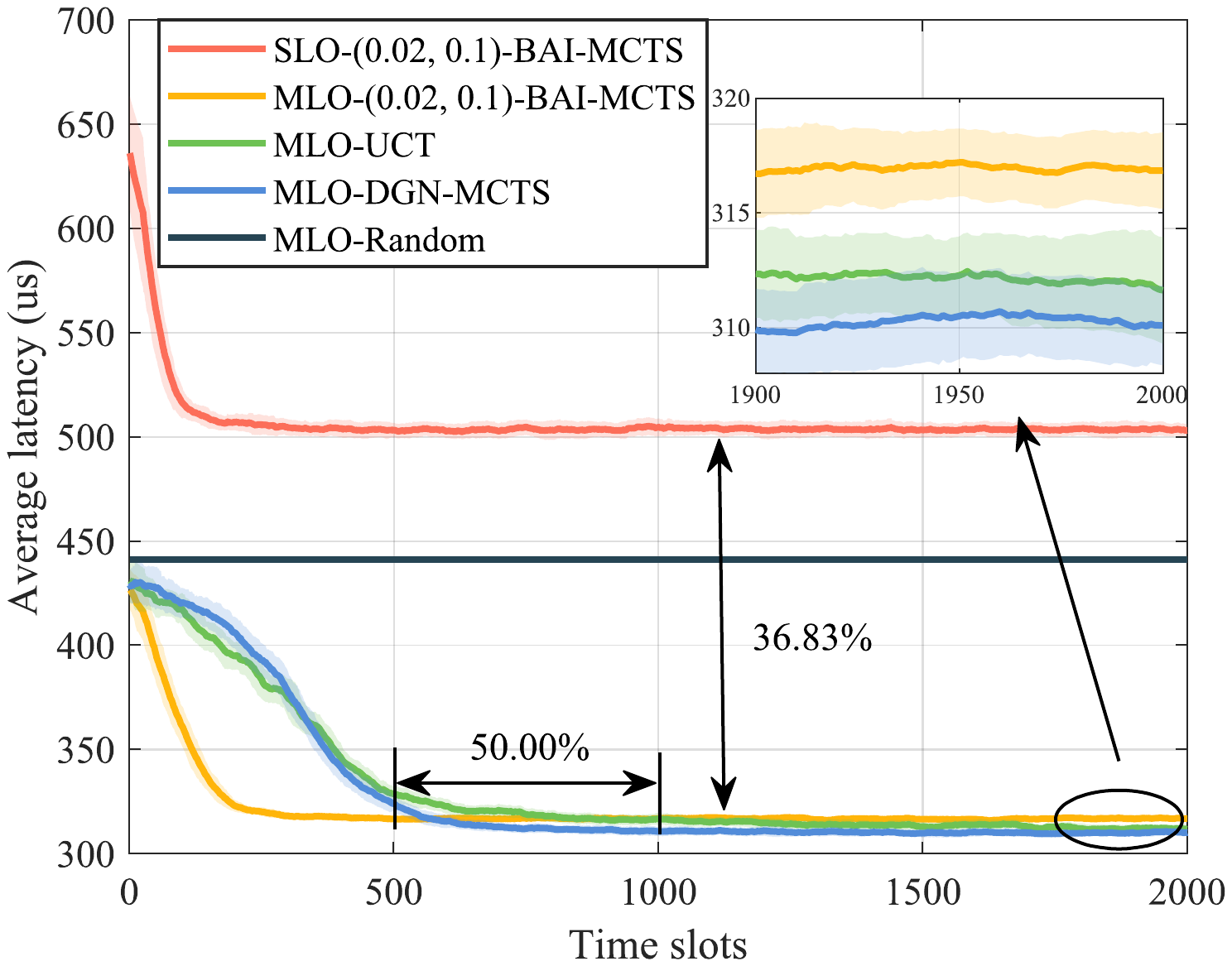}}
\caption{The average latency of the BAI-MCTS algorithm, the UCT algorithm, the DNG-MCTS algorithm, and the random selection method under the STR mode, and the BAI-MCTS algorithm under the SLO mode, where $T = 2,000$, $\epsilon = 0.02$, and $\delta=0.1$.}
\label{algorithm_L}
\end{figure}

In addition to throughput improvements, STR also achieves a significant performance in reducing latency, which is defined as the transmission time required for a packet to traverse from generation at the STA to successful reception at the AP.
Fig.~\ref{algorithm_L} depicts the average latency of different algorithms under the SLO and STR modes.
We see that the BAI-MCTS algorithm under the STR mode is better than that of the SLO mode, achieving a 36.83$\%$ reduction in latency.
This result underscores STR's superiority in optimizing delay-sensitive applications. In addition, the BAI-MCTS algorithm achieves a 50.00$\%$ faster convergence rate than the DGN-MCTS algorithm while maintaining a final latency disparity below $10$ us. This indicates that the BAI-MCTS algorithm can meet the practical latency requirements.

\begin{figure}[!t]
\centerline{\includegraphics[scale=0.29]{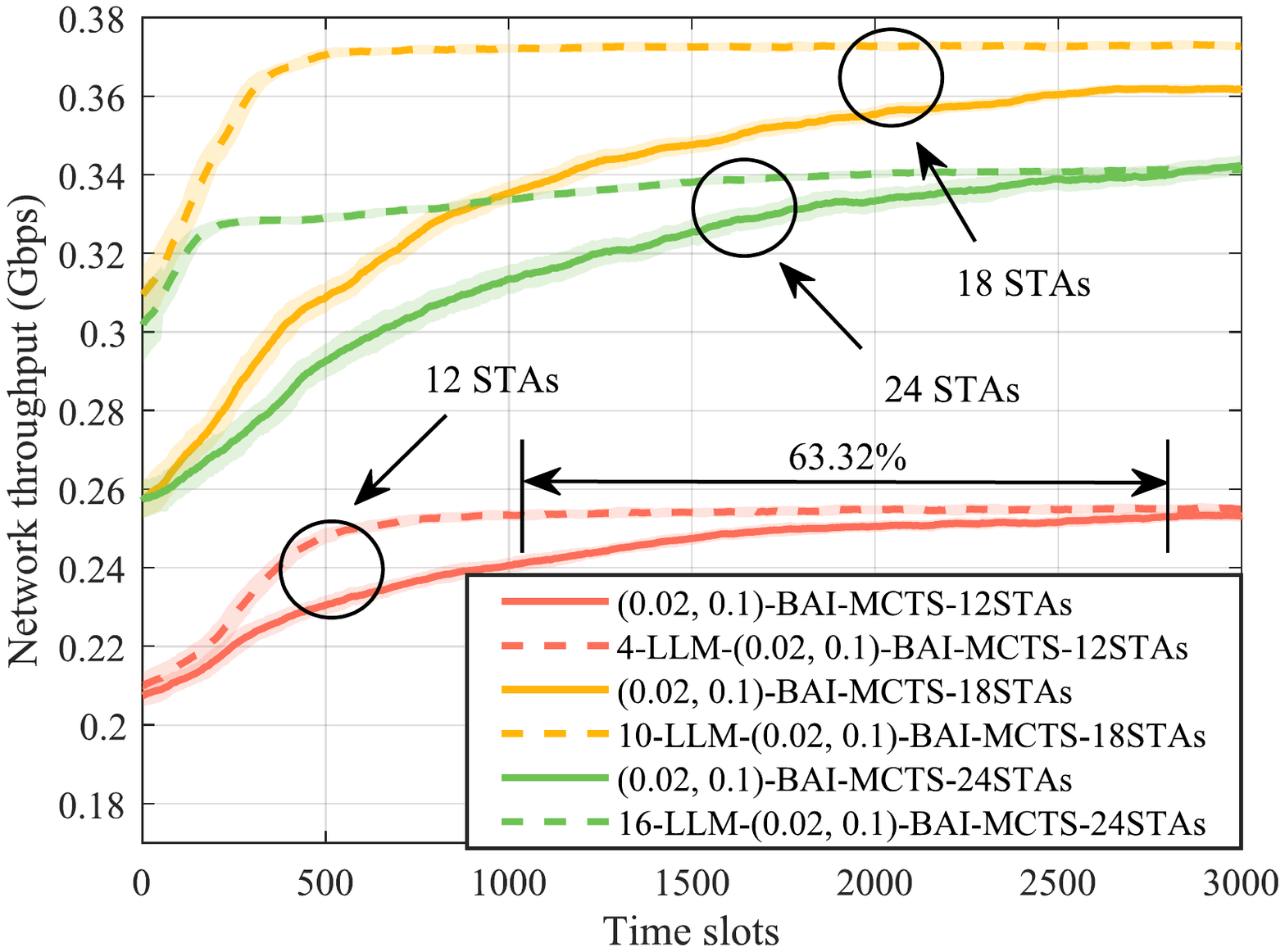}}
\caption{The average network throughput of the LLM-BAI-MCTS and BAI-MCTS algorithms when $\mathfrak{L} = N-8$, where $T = 3,000$ and $N_m = \{4,6,8\}$ for $m  \in \mathcal{M}$. The example provided in ICL is a network with 6 STAs and its channel allocation strategy.}
\label{STA12-24}
\end{figure}

Fig.~\ref{STA12-24} compares the performance of the BAI-MCTS and LLM-BAI-MCTS algorithms in denser networks with $N=\{12, 18, 24\}$ STAs. For the LLM-BAI-MCTS algorithm, we retain $\mathfrak{L} = N-8$ STAs allocated by the LLM. Meanwhile, the remaining 8 STAs are configured through the BAI-MCTS algorithm. We adopt the GPT-4o model for testing.
For ICL, we consider a scenario of $3$ APs and $6$ STAs, and the corresponding allocation strategy is obtained by solving problem \eqref{total}. 
From Fig.~\ref{STA12-24}, we see that the average network throughput of the LLM-BAI-MCTS algorithm is better than that of the BAI-MCTS algorithm across all scenarios. In addition, the BAI-MCTS algorithm requires more time slots to converge compared with the LLM-BAI-MCTS algorithm, resulting in high sample complexity. Compared with BAI-MCTS, the LLM-BAI-MCTS algorithm achieves a $63.32\%$ improvement in convergence rate in a network with $12$ STAs, and a more significant improvement in networks with $18$ and $24$ STAs. Thanks to the high-quality initialization, the sample complexity of the LLM-BAI-MCTS algorithm is not sensitive to the number of STAs. We also observe that the network throughput of $N=24$ is lower than that of $N=18$. This is because the denser the network, the more contention.

\begin{figure}[!t]
\centerline{\includegraphics[scale=0.31]{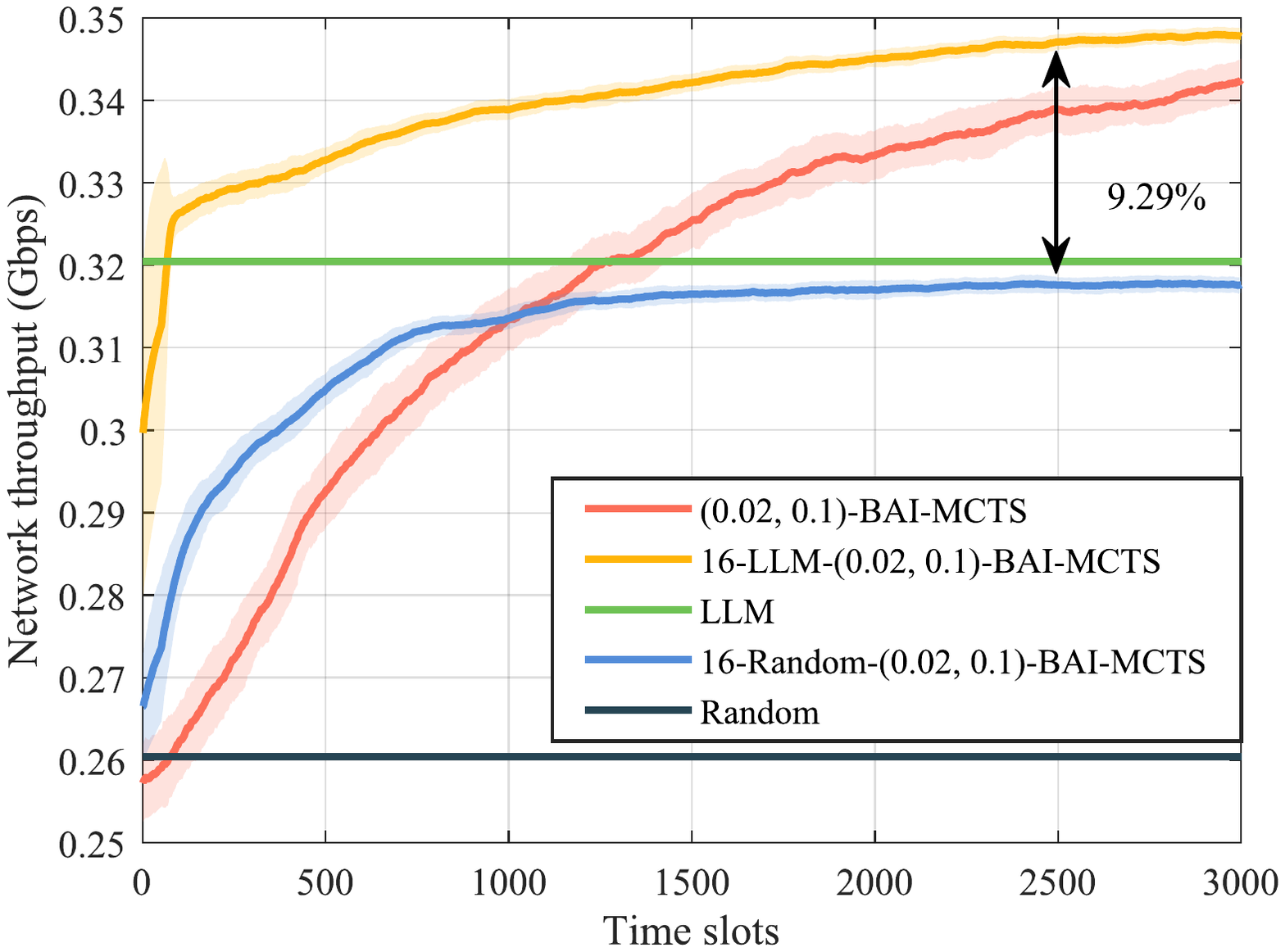}}
\caption{The average network throughput of different algorithms when $\mathfrak{L} = \{0,16,24\}$, where $T = 3,000$ and $N_m = 8$ for $m  \in \mathcal{M}$. The example provided in ICL is a network with $6$ STAs and its channel allocation strategy.}
\label{improvedalgorithm}
\end{figure}
We further investigate the proposed algorithms in a dense WiFi 7 network with $N = 24$ by comparing them against different algorithms. It is observed that the BAI-MCTS algorithm with 16 STAs using LLM exhibits a faster convergence rate compared to the pure BAI-MCTS algorithm. Additionally, the performance of the BAI-MCTS algorithm surpasses that of the method where all STAs use the LLM's output (i.e., $\mathfrak{L} = 24$). This suggests that relying solely on LLM is insufficient to effectively solve complex channel allocation problems. Notably, the performance of the BAI-MCTS algorithm with 16 STAs using LLM is approximately $9.29\%$ better than using random selection. These observations imply the superiority of integrating the LLM and the BAI-MCTS algorithm.
\begin{table}[!t] 
  \centering 
  \renewcommand\arraystretch{1.25}
  \caption{Average network throughput (Gbps) of the LLM-BAI-MCTS algorithm using different LLMs} 
  \label{LLMs} 
  \resizebox{0.40\textwidth}{!}{  
   \begin{tabular}{cccc}
    \hline
    \noalign{\hrule height 0.6pt} 
    LLMs & $\mathfrak{L} = 8$ & $\mathfrak{L} = 16$ & $\mathfrak{L} = 24$ \\
    \hline
    GPT-4o & \textbf{0.3337} & \textbf{0.3414} & \textbf{0.3205} \\
    Claude-3-Sonnet & 0.3324 & 0.3313 & 0.3202 \\
    Gemini-1.5-Pro & 0.3262 & 0.3287 & 0.3127 \\
    \hline
    \noalign{\hrule height 0.6pt} 
  \end{tabular}
  }
\end{table}

Finally, we evaluate the LLM-BAI-MCTS algorithm using different LLMs and parameter $\mathfrak{L}$. The considered LLMs are the GPT-4o, Claude-3-Sonnet, and Gemini-1.5-Pro models. The number of LLM-assisted STAs is set to $\mathfrak{L} = 8$, $16$, and $24$. The configurations are the same as those in Fig.~\ref{improvedalgorithm}. The average network throughput achieved by LLM-BAI-MCTS is summarized in Table~\ref{LLMs}. 
From the results, we observe that GPT-4o outperforms the other models in all scenarios, achieving the highest average network throughput. Notably, all three LLMs can address the task effectively, demonstrating that the LLM-BAI-MCTS algorithm possesses strong generalization ability and performs well in dense networks. Furthermore, the performance of different LLMs exhibits varying trends as $\mathfrak{L}$ changes, highlighting the tradeoff in selecting the optimal value of $\mathfrak{L}$. This observation reinforces the importance of balancing generalizability and performance through $\mathfrak{L}$.

\section{Conclusions and Discussions}\label{SecCon}
This paper studied the dynamic channel allocation problem in dense WiFi 7 networks with MLO by formulating it as an MAB problem. We proposed the BAI-MCTS algorithm to achieve efficient exploration-exploitation tradeoffs and derived theoretical guarantees on its performance. To further reduce sample complexity and enhance generalizability, we introduced the LLM-BAI-MCTS algorithm, which integrates LLMs through ICL and CoT techniques, providing high-quality initializations for faster convergence. Numerical results demonstrated that the BAI-MCTS algorithm improves convergence by $50.44\%$ over state-of-the-art methods, while the LLM-BAI-MCTS algorithm achieves an additional $63.32\%$ improvement in dense networks, showcasing strong scalability and adaptability. 

There are several limitations and future research directions in this work. First, this work focuses primarily on the STR mode. An interesting yet challenging problem is to extend the channel allocation problem to the NSTR mode or hybrid networks, which require sophisticated temporal alignment strategies and complex coordination mechanisms.  
Second, maximizing overall network throughput in problem \eqref{total} may result in resource allocation imbalance. 
An interesting problem is to explore alternative objective functions by modifying the objectives in the problem formulation, which would enhance practical applicability. Third, while the LLM-BAI-MCTS algorithm leverages general LLMs, further exploration of fine-tuning these models on domain-specific data could substantially improve performance, potentially reducing convergence time, especially in highly dynamic environments.
Last but not least, further expanding evaluations to real-world data or scenarios can effectively narrow the performance gap between simulation results and real-world outcomes.

{\appendices
\section{Proof of Theorem 1}\label{convergence}
At each layer of the MCT, the BAI-MCTS algorithm can be regarded as the $\text{EB-TC}_{\epsilon}$ algorithm. The heuristic properties of the MCTS algorithm can accelerate the exploration process in huge arm space. Since the BAI-MCTS algorithm performs the stopping rule at each layer, the upper bound on its sampling complexity can be obtained by summing up that at all layers. As a result, we can follow the layered structure of the BAI-MCTS algorithm to prove Theorem~\ref{Theorem1}. 

First, we derive an upper bound on the sampling complexity of the $(N-1)$-th layer. We assume that the algorithm has converged in the lower layers. In addition,  we consider that the nodes' rewards follow independent Gaussian distributions with values bounded within the interval [0, 1] (normalized network throughput) so that we can apply Theorem 1 in \cite{jourdan2023varepsilon}. The upper bound on the asymptotic expected sampling complexity for achieving ($\epsilon',\delta'$)-PAC in terms of its $|{\mathcal{L}}|$ child nodes is  
\begin{equation}\label{tau}  
\begin{aligned}  
\lim\sup_{\delta'\to0}\frac{\mathbb{E}_{\boldsymbol{\nu}}[\tau_{\epsilon',\delta'}^{N-1}(q)]}{\log(1/\delta')} \leq T_{\epsilon'}^{N-1}(\boldsymbol{\mu}_{q^{N-1}}),
\end{aligned}  
\end{equation} 
where $q=1,\ldots, |{\mathcal{L}}|^{N-1}$, and $\tau_{\epsilon',\delta'}^{N-1}(q)$ is the convergence time for the $q$-th node at the $(N-1)$-th layer.

For Gaussian bandits, the characteristic times can be obtained by solving a simpler optimization problem \cite{jourdan2023varepsilon}. Assume that $d^*(\boldsymbol{\mu})=\{d^*\}$. Let $r(\boldsymbol{\mu})$ be the solution of $\psi_{\boldsymbol{\mu},\epsilon}(r)=0$, for all $r\in(1/\min_{d\neq d^*}(\mu_{d^*}-\mu_d+\epsilon)^2,+\infty)$. Then, we have
\begin{equation}\psi_{\boldsymbol{\mu},\epsilon}(r)=\sum_{d\neq d^*}\frac1{(r(\mu_d^*-\mu_d+\epsilon)^2-1)^2}-1,
\end{equation}
and
\begin{equation}
T_\epsilon(\boldsymbol{\mu})=\frac{2r(\boldsymbol{\mu})}{1+\sum_{d\neq d^*}\frac1{r(\boldsymbol{\mu})(\mu_d^*-\mu_d+\epsilon)^2-1}},
\end{equation}
where $\psi_{\boldsymbol{\mu},\epsilon}(r)$ is a convex and decreasing function.

Next, we generalize the result \eqref{tau} to all layers. A similar proof is conducted from layer $h$ to $h-1$ for all $1 \leq h \leq {N-1}$. We use Lemma 23 from \cite{jourdan2023varepsilon}, which states:
 
\textit{Let $\gamma > 0$. There exists $T_1(\boldsymbol{\mu})$ with $\mathbb{E}_{\boldsymbol{\nu}}[T_1] < +\infty$ such that for all $t \geq T_1(\boldsymbol{\mu})$ and all $d \neq d^*$,}
\begin{equation}
\frac{\mathbb{N}_{t,d}}{\mathbb{N}_{t,d^*}} \leq \frac{\boldsymbol{w}^*(d)}{\boldsymbol{w}^*(d^*)} + \gamma,
\end{equation}
\textit{where $\boldsymbol{w}^{*}(d)$ is the optimal allocation proportion for node $d$}.
Using this lemma, we can establish that, for sufficiently large $t$, the reward of the node of layer $(h-1)$ can be approximated by weighting the rewards of its child nodes, which follow independently and identically Gaussian distributions, using the corresponding optimal allocation vector $\boldsymbol{w}^*$.
Therefore, the rewards of nodes at layer $h-1$ also follow Gaussian distributions and are bounded within the interval [0, 1]. This allows us to apply the result \eqref{tau} recursively to all layers.

Based on this layer-by-layer convergence analysis of the BAI-MCTS algorithm, we ensure that each layer is ($\epsilon',\delta'$)-PAC. Let $\epsilon' = {\epsilon}/{N}$ and $\delta' = 1 - \sqrt[N]{1 - \delta}$. After all layers converge, there is a probability greater than $1 - \delta$ to select an $\epsilon$-optimal arm. Here, $1 - \delta$ is obtained through $\left(1 - (1 - \sqrt[N]{1 - \delta})\right)^{N}$, and $\epsilon$ is derived from $\left({\epsilon}/{N} \times N\right)$. At this point, we derive the upper bound on the total sampling complexity of the BAI-MCTS algorithm, which is given by 
	\begin{equation}
		\begin{aligned}
  &\lim\sup_{\delta\to0}\frac{\mathbb{E}_{{\boldsymbol{\nu}}}[\tau_{\epsilon,\delta}]}{\log(1/\delta)}\leq \sum_{h=0}^{N-1}{\lim\sup_{\delta\to0}\frac{\mathbb{E}_{{\boldsymbol{\nu}}}\left[\tau_{\frac{\epsilon}{N},1-\sqrt[N]{1-\delta}}^{h}\right]}{\log(1/(1-\sqrt[N]{1-\delta}))}}\\
   &\leq {\sum_{h=0}^{N-1}\max_{q=1,\ldots,|{\mathcal{L}}|^{h}} T_\frac{\epsilon}{N}^{h}(\boldsymbol{\mu}_{q^h})}
   \leq \sum_{h=0}^{N-1}{\hat T_\frac{\epsilon}{N}^{h}}(\boldsymbol{\mu}^{h}),
		\end{aligned}
	\end{equation}
where $\tau_{\epsilon,\delta}^h$ is the convergence time at layer $h$.

This concludes the proof.
$\hfill\blacksquare$

\section{Proof of Theorem 2}\label{error}
In Algorithm \ref{BAI-MCTS}, the indicator $\eta$ records the number of convergence layers. Hence, the layers whose numbers are at or below $\eta$ are converged, while those above $\eta$ are not. This motivates us to divide the proof into two parts. 

First, we consider the layers that have not yet converged. We define some key concepts and notations as in Theorem 6  \cite{jourdan2023varepsilon}. The number of different reward means is defined as $C_{\boldsymbol{\mu}}:=|\{\mu_d\mid d\in \mathcal{A}(\boldsymbol{\mu})\}|$ in terms of $\boldsymbol{\mu}$. The
arms with the same suboptimal gap are defined as $\mathscr{C}_{\boldsymbol{\mu}}(d):=\{d\in \mathcal{A}(\boldsymbol{\mu})\mid\mu_{d^*}-\mu_d=\Delta_d\}$, which is ordered in ascend order  $0=\Delta_1<\Delta_2<\ldots<\Delta_{C_{\boldsymbol{\mu}}}\:$.
In particular, we have $\mathscr{C}_{\boldsymbol{\mu}}(1)=d^*(\boldsymbol{\mu})$ and $\Delta_{C_{\boldsymbol{\mu}}}=\Delta_\mathrm{max}$. Then, we can obtain the following result.

\textit{
Let $\epsilon_0 > 0$. For all $\epsilon \geq 0$, if $\epsilon \in [\Delta_d, \Delta_{d+1})$ and $d \in [C_{\boldsymbol{\mu}} - 1]$, let $d_{\boldsymbol{\mu}}(\epsilon) = d$; Otherwise, $d_{\boldsymbol{\mu}}(\epsilon) = C_{\boldsymbol{\mu}}$. For all $d \in [C_{\boldsymbol{\mu}} - 1]$, let $C_{d}(\epsilon_{0})=2\Delta_{d}^{-1}-\epsilon_{0}^{-1}$ and $C_{d,j}(\epsilon_{0})=2\frac{\Delta_{j}/\epsilon_{0}+1}{\Delta_{d}-\Delta_{j}}+3\epsilon_{0}^{-1}$. For all $d \in [C_{\boldsymbol{\mu}} - 1]$, let $H_d(\boldsymbol{\mu}, \epsilon_0) := \min_{j \in [d]} \max \{\bar{H}_{d,j}(\boldsymbol{\mu},\epsilon_0), \tilde{H}_{d,j}(\boldsymbol{\mu},\epsilon_0)\}$, where}
\begin{equation}
\begin{aligned}
&\bar{H}_{d,j}(\boldsymbol{\mu},\epsilon_{0}) :=|d^*(\boldsymbol{\mu})|\max\left\{\sqrt{2}\Delta_{j+1}^{-1}, C_{d+1,j}(\epsilon_0)\right\}^2+ \\
&\max\left\{C_{j+1}(\epsilon_0), C_{d+1,j}(\epsilon_0)\right\}^2 ( \sum_{k=2}^j |\mathscr{C}_{\boldsymbol{\mu}}(k)| + \sum_{k=d+1}^{C_{\boldsymbol{\mu}}} |\mathscr{C}_{\boldsymbol{\mu}}(k)| ) \\
&+\sum_{k=j+1}^d|\mathscr{C}_{\boldsymbol{\mu}}(k)|\max\left\{C_{j+1}(\epsilon_0), C_{d+1,j}(\epsilon_0),\sqrt{2}\Delta_k^{-1}\right\}^2 , \\
\end{aligned} 
\end{equation}
\begin{equation}
\begin{aligned}
\tilde{H}_{d,j}(\boldsymbol{\mu},\epsilon_0) :=& \frac{2|d^*(\boldsymbol{\mu})|}{\Delta_{j+1}^2}+\frac{2\sum_{k=1}^j|\mathscr{C}_{\boldsymbol{\mu}}(k)|}{(\Delta_{d+1}-\Delta_j)^2}+\\
&\sum_{k=2}^j|\mathscr{C}_{\boldsymbol{\mu}}(k)|\max\left\{C_{j+1}(\epsilon_0),\epsilon_0^{-1}\right\}^2+ \\
&\sum_{k=j+1}^{C_{\boldsymbol{\mu}}}|\mathcal{C}_{\boldsymbol{\mu}}(k)|\max\{C_{j+1}(\epsilon_{0}),\epsilon_{0}^{-1},\sqrt{2}\Delta_{k}^{-1}\}^{2}.
\end{aligned} 
\end{equation}
\textit{For the EB-TC$_{\epsilon_0}$ algorithm with a fixed proportion $\beta = 1/2$ and $\epsilon\geq0$, when the algorithm has not stopped yet, for ${\boldsymbol{\nu}} \sim\mathfrak{N}_{{|\mathcal{L}|}}$ and $t \geq D_{\boldsymbol{\mu}}$, it holds that
\begin{equation}
\begin{aligned}
\mathbb{P}_{{\boldsymbol{\nu}}}\left(\hat{\mathcal{I}}_t\notin\mathcal{A}_{\epsilon}(\boldsymbol{\mu})\right)\leq&\mathds{1}\left(\epsilon<\Delta_{\max}\right)\frac{|\mathcal{L}|(|\mathcal{L}|+1)}{2}e^{2}\\
&(2+\log t)^{2}p\left(\frac{t-5|\mathcal{L}|^{2}/2}{8H_{d_{\boldsymbol{\mu}(\epsilon)}}(\boldsymbol{\mu},\epsilon_{0})}\right),
\end{aligned} 
\end{equation}
where $p(x) = xe^{-x}$ and $D_{\boldsymbol{\mu}}=8H_1(\boldsymbol{\mu},\epsilon_0)
h_2(8H_1(\boldsymbol{\mu},\epsilon_0),$\\
$5|\mathcal{L}|^2/2,2+\log\left(|\mathcal{L}|(|\mathcal{L}|+1)/2)\right)+5|\mathcal{L}|^2/2$ with $h_2(x,y,z)=\inf\left\{u\mid u-\log u-2\log\left(2+\log(xu+y)\right)\geq z\right\}.$}

Consequently, for layer $\eta$, there is a probability greater than $Q(|\mathcal{L}|,\mathbb{N}_{t,d_{t,\eta}},{\boldsymbol{\mu}}_{d_{t,\eta}},{\epsilon}/{N})$ to choose an ${\epsilon}/{N}$-optimal arm when $\mathbb{N}_{t,d_{t,\eta}} \geq D_{{\boldsymbol{\mu}}_{d_{t,\eta}}}$ and $\epsilon_0=\epsilon$, where 
\begin{equation}\label{Q}
Q(K,t,\boldsymbol{\mu},\epsilon)= \frac{K(K+1)}{2}e^{2}(2+\log t)^{2}p\left(\frac{t-5K^{2}/2}{8H_{d_{\boldsymbol{\mu}(\epsilon)}}(\boldsymbol{\mu},\epsilon)}\right).
\end{equation}
By leveraging Lemma 23 in \cite{jourdan2023varepsilon}, we further extend this analysis to subsequent layers. When $\mathbb{N}_{t,d_{t,h}^*} \geq \max\{D_{{\boldsymbol{\mu}}_{d_{t,h}^*}}, T_1({{\boldsymbol{\mu}}_{d_{t,h}^*}})\}$, $\forall h=\eta+1,\ldots, N-1$, we have a probability greater than
\begin{equation}\label{with_w}
\begin{aligned}
\prod_{h=\eta+1}^{N-1} \boldsymbol{w}_{d_{t,h-1}^*}^*(d_{t,h}^*) \left( 1 - Q\left( |\mathcal{L}|, \mathbb{N}_{t,d_{t,h}^*}, \boldsymbol{\mu}_{d_{t,h}^*}, \frac{\epsilon}{N} \right) \right) 
\end{aligned}
\end{equation}
to find an ${\epsilon}/{N}$-optimal arm, where ${d_{t,h}^*}$ is the best arm at layer $h$ following the converged nodes at time slot $t$. Furthermore, it can be  scaled as (Lemma 20 in \cite{jourdan2023varepsilon})
\begin{equation}
\begin{aligned}
&\prod_{h=\eta+1}^{\mathrm{N}-1} \frac{1}{4\sqrt{2(|\mathcal{L}|-1)}} \left( 1 - Q\left( |\mathcal{L}|, \mathbb{N}_{t,d_{t,h}^*}, \boldsymbol{\mu}_{d_{t,h}^*}, \frac{\epsilon}{N} \right) \right) .
\end{aligned}
\end{equation}

Second, we analyze the converged layers. According to the properties of the BAI-MCTS algorithm in Theorem~\ref{Theorem1}, we know that there is a probability greater than 
\begin{equation}
(1-(1-\sqrt[N]{1-\delta}))^{\eta} =  (1-\delta)^{\frac{\eta}{N}}
\end{equation}
for all the converged layers to find the ${\epsilon}/{N}$-optimal arms.

Finally, we integrate the results from both the non-converged and converged layers to derive a comprehensive bound. At time slot $t$, for ${\boldsymbol{\nu}} \sim \mathfrak{N}_{{|\mathcal{L}|}^{N}}$, $\mathbb{N}_{t,d_{t,\eta}} \geq D_{{\boldsymbol{\mu}}_{d_{t,\eta}}}$, and $\mathbb{N}_{t,d_{t,h}^*} \geq \max\{D_{{\boldsymbol{\mu}}_{d_{t,h}^*}}, T_1({{\boldsymbol{\mu}}_{d_{t,h}^*}})\},\forall h=\eta+1,\ldots, N-1$, if the BAI-MCTS algorithm with a fixed proportion $\beta = 1/2$ has not stopped yet, then $\forall\epsilon\geq0$,
\begin{equation}
\begin{aligned}
&\mathbb{P}_{\boldsymbol{\nu}}\left(\hat{\mathcal{I}}_t \notin \mathcal{A}_{\epsilon}(\boldsymbol{\mu})\right) \leq \mathds{1}\left(\epsilon < \Delta_{\mathrm{max}}\right)\left(1 - (1-\delta)^{\frac{\eta}{N}} \right) \\
&\times \prod_{h=\eta+1}^{N-1} \frac{1}{4\sqrt{2(|\mathcal{L}| - 1)}} \left(1 - Q\left(|\mathcal{L}|, \mathbb{N}_{t, d_{t, h}^*}, \boldsymbol{\mu}_{d_{t, h}^*}, \frac{\epsilon}{N}\right)\right) \\
&\times \left(1 - Q\left(|\mathcal{L}|, \mathbb{N}_{t, d_{t, \eta}}, \boldsymbol{\mu}_{d_{t, \eta}}, \frac{\epsilon}{N}\right)\right).
\end{aligned}
\end{equation}

This concludes the proof.
$\hfill\blacksquare$}

\bibliography{Ref_MLO}
\bibliographystyle{IEEEtran}

\end{document}